\documentclass[12pt]{amsart}

\usepackage{setspace,color}
\usepackage{graphicx,boxedminipage}
\usepackage{url}
\usepackage{xcolor}
\usepackage{amsmath,amssymb}
\usepackage{booktabs}
\usepackage{amsaddr}

\newtheorem{Theorem}{Theorem}[section]
\newtheorem{Lemma}[Theorem]{Lemma}

\newtheorem{Remark}[Theorem]{Remark}

\numberwithin{equation}{section}

\newcommand{\eee}{{\rm e}}
\newcommand{\ddd}{{\rm d}}
\allowdisplaybreaks
\begin{document}

\setcounter{page}{0} \renewcommand{\thefootnote}{\fnsymbol{footnote}} %
\thispagestyle{empty} \centerline{ }

\vskip 0.1 true cm
\begin{center}
{\Large \textsf{A Unifying Theory of Aging and Mortality}}

\vskip 1 true cm \noindent

{\sc Valentin Flietner}\textsuperscript{1}

{\sc Bernd Heidergott}\textsuperscript{2}

{\sc Frank den Hollander}\textsuperscript{3}

{\sc Ines Lindner}\textsuperscript{4}

{\sc Azadeh Parvaneh}\textsuperscript{5}

{\sc Holger Strulik}\textsuperscript{6}

\vskip 0.5 true cm

\parbox{15cm}{\singlespacing \footnotesize \textsuperscript{1} PwC, Bernhard-Wicki-Strasse 8, 80636 Munich, Germany; \email{valentin@flietner.com}

\textsuperscript{2} VU University Amsterdam, Department of Operations Analytics, De Boelelaan 1105, 1081 HV Amsterdam, The Netherlands; \email{bheidergott@feweb.vu.nl}

\textsuperscript{3} Leiden University, Mathematical Institute, Einsteinweg 55, 2333 CC Leiden, The Netherlands; \email{denholla@math.leidenuniv.nl}

\textsuperscript{4} VU University Amsterdam, Department of Economics, De Boelelaan 1105, 1081 HV Amsterdam, The Netherlands; \email{i.d.lindner@vu.nl}

\textsuperscript{5} Leiden University, Mathematical Institute, Einsteinweg 55, 2333 CC Leiden, The Netherlands; \email{s.a.parvaneh.ziabari@math.leidenuniv.nl}

\textsuperscript{6} University of G{\"ottingen}, Department of Economics, Platz der G{\"o}ttinger Sieben 3, 37073 G{\"ottingen}, Germany; Correspondence to holger.strulik@wiwi.uni-goettingen.de.}

\vskip 1 cm
\noindent
\begin{minipage}{15cm}
\textbf{Abstract}. In this paper, we advance the network theory of aging and mortality by developing a causal mathematical model for the mortality rate. First, we show that in large networks, where health deficits accumulate at nodes representing health indicators, the modeling of network evolution with Poisson processes is universal and can be derived from fundamental principles. Second, with the help of two simplifying approximations, which we refer to as mean-field assumption and homogeneity assumption, we provide an analytical derivation of Gompertz law under generic and biologically relevant conditions. We identify the parameters in Gompertz law as a function of the parameters driving the evolution of the network, and illustrate our computations with simulations and analytic approximations.

\end{minipage}

\end{center}
\renewcommand{\thefootnote}{\arabic{footnote}} \setcounter{footnote}{0} %
\setcounter{page}{0}

\clearpage


\section{Background and challenges}


\subsection{Motivation}

Why do we die when we get old? Traditionally, it is understood that we eventually die of `old age'. In scientific terms, the mortality rate $m(t)$ at age $t$ describes the probability that a person dies in a short age interval following age $t$, namely, for $0 < \Delta \ll 1$ the probability of death during the time interval $[t,t+\Delta)$ given that no death occurred during the time interval $[0,t]$ equals $m(t)\Delta$. The seminal paper by Mitnitski, Rutenberg, Farrell and Rockwood \cite{MRFR2017} led to a paradigm shift in the interpretation of the age-dependency of the mortality rate (see also \cite{Farrell2016, Rutenberg2018, Taneja2016}). The mortality rate is modeled by a \emph{dynamic network}, in which nodes represent health indicators, the dynamic relationship between the states of the nodes models the interdependency of the health indicators, and death is defined as the time it takes to reach a network state in which two carefully selected nodes, called \emph{mortality nodes}, are both damaged. The model showed a good fit to actual mortality rate data and offered a new interpretation of death as the accumulation of damaged health indicators. However, the existing network theory of aging and mortality depends on identifying an appropriate network structure through numerical experiments in a trial-and-error process.

In this study, we move from simulation-based modeling to a robust and foundational framework and derive a causal mathematical model for the mortality rate under generic conditions. Our holy grail is to provide a mathematical derivation of \emph{Gompertz law} for $m(t)$, which is believed to be valid when $t$ is neither too small nor too large. Before we state the mathematical model and identify our research question, we discuss in more depth the interpretation of the concept of `node' and `state of a node'. Using results from network science, we explain why a certain network structure and a certain network dynamics can be argued for in the model.

The results of our study advance the theoretical understanding of the age-dependency of the mortality rate. Our analytical approach relies solely on very general structural assumptions about the network and supports the network choice made in \cite{MRFR2017}. This means that results that were previously derived from an ad-hoc fitting of a particular network structure are here derived as closed-form analytical solutions under broad and biologically relevant conditions.

\subsection{Gompertz law}

Human life can be roughly divided into two periods: a {\it phase of initial development} (ranging from birth to puberty), during which the mortality rate decreases, followed by a {\it phase of aging} (after puberty), during which the mortality rate increases. Figure \ref{fig:USmen} illustrates these two phases for US men in the period 2010-2019. It indicates that the phase of aging runs from $10$ to $100$ years of age, and that from $40$ years of age onwards the mortality rate and age are log-linearly related. This empirical relationship, first observed and stated in 1825 by actuary Benjamin Gompertz \cite{G1825}, is given by ($\approx$ means approximately)
\begin{equation}
m(t) \approx \alpha\, \eee^{\beta t}, \qquad t \geq 0,
\label{Gomp0}
\end{equation}
with parameters $\alpha,\beta>0$, and is since referred to as \emph{Gompertz law}.

\begin{figure}[htbp]
\centering
\vspace{-1cm}
\includegraphics[width=0.5\textwidth, angle =-90]{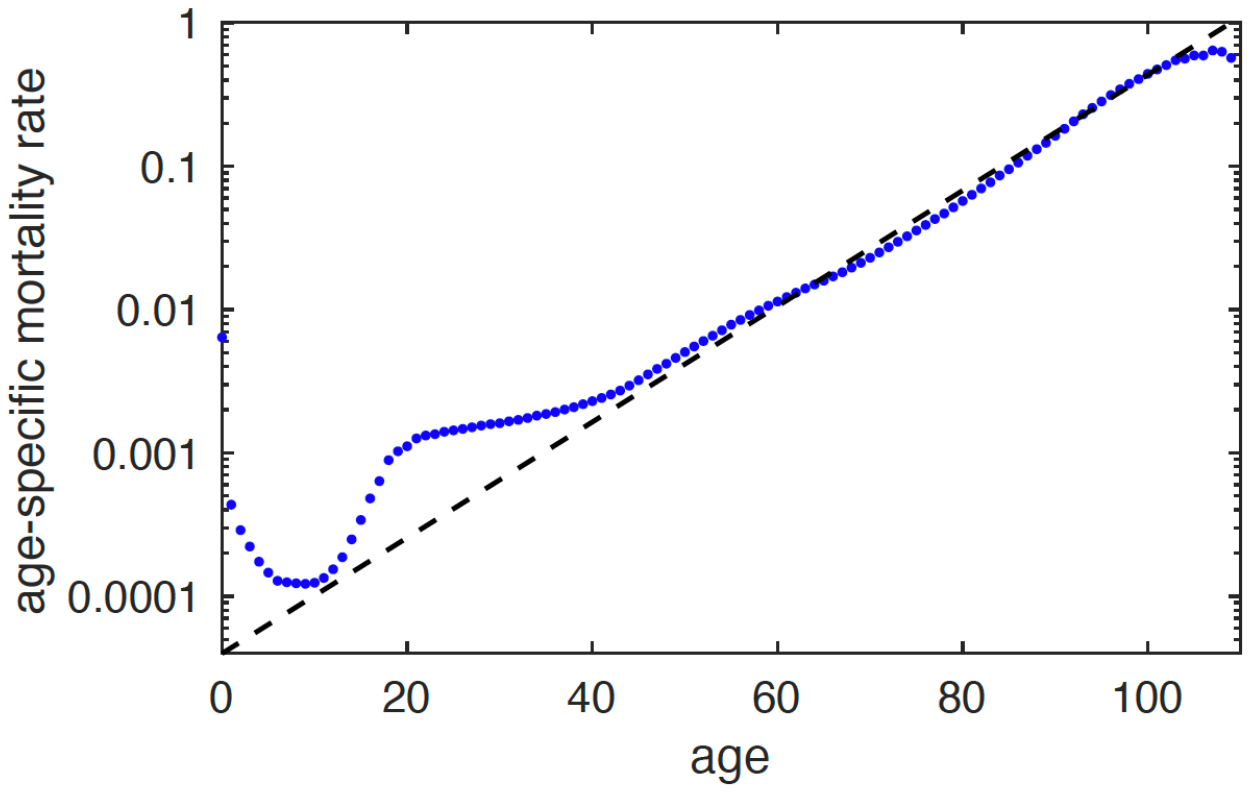}
\vspace{-1.5cm}
\caption{\small Age-specific mortality rate for US men in the period 2010-2019. Dots indicate data points taken from \cite{HMD2024}, \cite{HMDalt}. The straight line shows the Gompertz estimate $m(\text{\rm age}) = \alpha\,\eee^{\beta \times \text{\rm age}}$ with $\alpha \approx 5.8 \times 10^{-5}$ and $\beta \approx 8.5 \times 10^{-2}$, and with age measured in years. The plot is semi-logarithmic. (The data beyond 100 years may be less reliable because of measurement errors and selection effects.)}
\label{fig:USmen}
\end{figure}


Figure \ref{fig:USmen} illustrates Gompertz law by showing it together with the mortality rate of US men in the period 2010-2019. The figure indicates two \emph{overshooting phases} relative to the Gompertz law, one between birth and $10$ years of age and one between $10$ and $40$ years of age. Newborns are more likely to die of diseases that are far from fatal for grown-ups, such as gastroenteritis. An explanation is offered by health deficit models due to a lack of redundancy, such as organ reserves \cite{DHS2021, SY2001}. The lowest mortality is reached in puberty, when the body is fully developed and aging starts to take over from growth. From a mere physiological point of view, there is no explanation for the second overshooting phase, after puberty. Rather, this peak likely stems from other influences, such as risk preference, suicide, or war \cite{S2021}. Note that \eqref{Gomp0} shows no barrier for increasing age, and hence refutes the frequently made assumption of a bounded life time. The fact that the sample size of all people who ever lived on earth increases continuously simply implies that the maximum ever-observed life length will keep on rising as time proceeds \cite{FP1996}.

The Gompertz parameters differ across sexes and across countries. On average across countries, women face a lower $\alpha$ and a higher $\beta$ \cite{GG1991}, i.e., women have an initial advantage of a lower rate of aging that men eventually catch up with when getting older. Economically advanced countries are characterised by a lower $\alpha$ and a higher $\beta$, i.e., a lower initial mortality rate and a faster speed of aging \cite{SV2013}. A lower $\alpha$ can be associated with better (initial) physiological conditions, such as nutrition. The time invariance of the Gompertz law suggests that, correcting for country-specific and sex-specific background risk, humans share a common mechanism of aging, a common stochastic process according to which individual bodies lose function over time and bodily failures and health deficits accumulate.

\subsection{Features of health in an aging network}

Medicine as a discipline is almost as old as mankind itself. Even though over the centuries medicine has undergone major paradigm shifts, a body of knowledge has emerged on how to assess the health status of a person, based on a few key health indicators such as blood pressure, sugar and cholesterol levels, heart rate, physical mobility, etc. With the advent of modern technology, more and more health indicators have become available. While such health indicators increase the knowledge on the health state, their added value is limited when measurements become more and more specific. Physicians need a small number of measurements to act on and, through centuries of medical practice, a few \emph{central measurements} have emerged. The fact that these measurements represent the health state of the body fairly well indicates their predictive power.

Modelling the health indicators as \emph{nodes} in a graph and their interconnectedness as \emph{links} in the graph motivates us to assume that the network is \emph{scale-free}, i.e., its empirical degree distribution approximately follows a power law. Furthermore, the \emph{hubs} of the networks should be separated from each other, so that they can serve as a proxy of the network state in the local clusters that they dominate. This makes it reasonable to suppose that the network is \emph{disassortative}, i.e., its degrees tend to be negatively correlated. We distinguish between \emph{mortality nodes} and \emph{aging nodes}. The former (small in number) are the health indicators that play a dominant role in the cause of death, while the latter (large in number) are the health indicators that form the overall network structure. Each node can be in a \emph{healthy state} or a \emph{damaged state}.

The model allows for a more speculative interpretation by taking the graph as a representation of the interconnectedness of physical parts of the body. In this interpretation, the mortality nodes represent the most vital organs, and the graph structure reflects the fact that certain physiological aspects of the body are more related to certain vital organs than others. Dissortativeness comes into play as a way of separating the different vital hubs as much as possible, so that if one hub is in a state of bad health, then its illness is contained as much as possible. The similarity with security structures of computer networks, or protocols for containing spread of diseases in infection models, is obvious.

The following assumptions (\emph{A}) and questions (\emph{Q}) serve as guidelines:
\begin{itemize}
\item
\emph{A:} The aging nodes and the mortality nodes serve as predictors of the overall state of the network.\\
\emph{Q:} Can we show that when the mortality nodes reach the damaged state, a certain fraction of the aging nodes is damaged as well. Can we estimate this fraction?
\item
\emph{A:} Scale-freeness and disassortativity are essential and rely on a hierarchy of health indicators.\\
\emph{Q:} To what extent is the presence of hubs crucial for predictability? If we decrease the out-degree of the hubs, then does the predictability of these hubs decrease?
\end{itemize}
The goal of the present paper is to propose a \emph{causal} model that allows for a \emph{data science approach} to mortality. In Section~\ref{sec:basic} we define the basic model. In Section~\ref{sec:analysis} we present our analysis of this model. In Section~\ref{sec:simulations} we illustrate the results with simulations and in Section 5 we provide an analytical approximation of Gompertz law. Section 6 concludes.


\section{Basic model}
\label{sec:basic}


\subsection{Aging-network process}

We consider a graph $G$ consisting of $n$ nodes and certain links between these nodes, such that for each node $i$ the set $N(i)$ of neighbours of $i$ is non-empty. Each node has a state that is either $0$ (= healthy) or $1$ (= damaged). A node represents a data measurement of the body. It is conceivable that there are limitless measurement possibilities for the body and that these have some causal interference structure. Since this structure is not known, we replace it by a graph. In Figure~\ref{fig:network} this graph is depicted, where the black nodes represent the aging nodes and the red nodes the mortality nodes.

\begin{figure}[htbp]
\centering
\vspace{-0.3cm}
\includegraphics[width=0.4\textwidth]{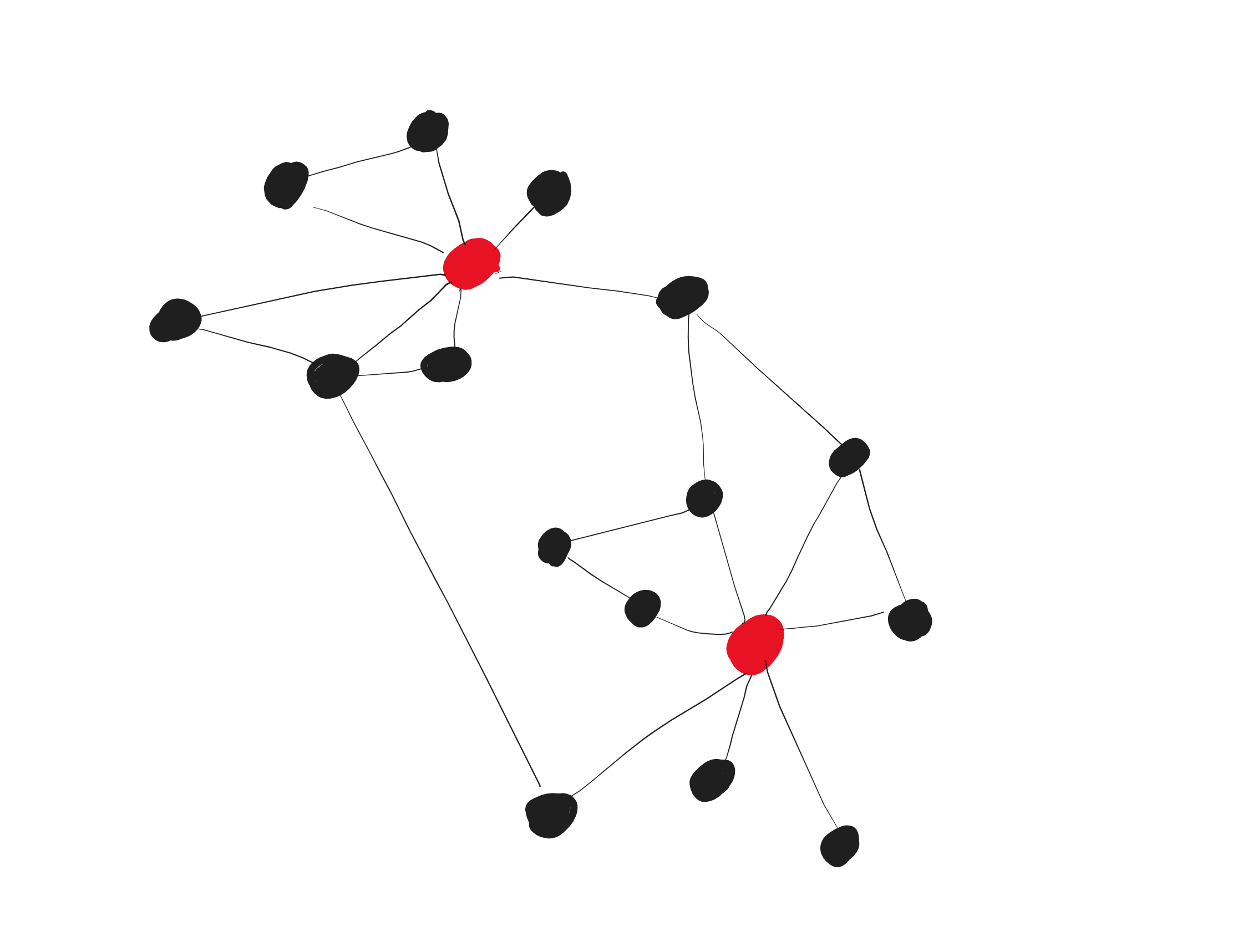}
\caption{\small An example of a health network. The aging nodes are black, the mortality nodes are red.}
\label{fig:network}
\end{figure}

We consider a Markov process ${\bf Z}=({\bf Z}(t))_{t \geq 0}$, called the \emph{aging process} on $G$, with
\[
{\bf Z}(t) = (Z_1(t ),\ldots,Z_n(t))
\]
the network state at time $t$, where $Z_i(t) \in \{0,1\}$ is the state of node $i$ at time $ t$. The evolution of the nodes is as follows. At time $t$, node $i$ goes from the healthy state to the damaged state, and vice versa, at rates, respectively,
\begin{equation}
\Gamma_+ (i,t) = A_+(f_i(t)),  \qquad \Gamma_- (i,t) = A_-(f_i(t)),
\label{expdamheal}
\end{equation}
where
\begin{equation}
\label{eq:frailty}
f_i (t) = \frac{1}{|N(i)|} \sum_{j \in N (i)} 1_{Z_j(t) = 1},
\end{equation}
is the fraction of damaged neighbours of node $i$ at time $t$, and
\begin{equation}
\label{eq:Apmdef}
A_+(f) = \Gamma_0\,\eee^{r_+ f},  \qquad A_-(f) = \frac{\Gamma_0}{R} \,\eee^{-r_- f}, \qquad f \in [0,1],
\end{equation}
are two functions that play a key role throughout the sequel. Here, $r_+,r_-,\Gamma_0,R \in (0,\infty)$ play the role of \emph{tuning parameters}. The exponential forms in \eqref{expdamheal} will be motivated in Section~\ref{sec:netmot}. Think of $\Gamma_0$ as the evolution rate for the network as a whole, and of $R$ as tuning an asymmetry between the healthy state and the damaged state. Note that ${\bf Z}$ has \emph{time-dependent} transition rates, i.e., it is a \emph{time-inhomogeneous Markov process}.

We label \emph{mortality nodes} by $1,2$, the \emph{aging nodes} by $3,\dots,n$, and define
\[
\begin{aligned}
S_{\rm{predeath}} &= \big\{z \in \{0,1\}^n\colon\,(z_1,z_2) \in \{(1,0),(0,1)\}\big\},\\
S_{\rm{death}} &= \big\{z \in \{0,1 \}^n\colon\,(z_1,z_2)=(1,1)\big\},\\
S_{\rm{other}} & =  \{0,1 \}^n  \setminus ( S_{\rm{predeath}} \cup S_{\rm{death}} )
=  \big\{z \in \{0,1\}^n\colon\,(z_1,z_2) =  (0,0) \big\} .
\end{aligned}
\]
Initially, all nodes are healthy, i.e.,
\[
{\bf Z}(0) = (0,\ldots,0),
\]
and the aging process ${\bf Z}$ can move into and out of \emph{every} state. The life time of the individual terminates when ${\bf Z}$ enters $S_{\rm{death}}$, i.e., at time
\begin{equation}
\label{eq:deathtime}
\tau = \inf\{t \geq 0\colon\,{\bf Z}(t) \in S_{\rm{death}}\}.
\end{equation}


\subsection{Mortality rate}

We are interested in the \emph{mortality rate} at time $t$ given by
\begin{equation}
\label{eq:mortality}
m(t) = \lim_{\Delta \downarrow 0} \frac{1}{\Delta} \mathbb{P}(\tau \leq t+\Delta \mid  \tau \geq t)
= - \frac{1}{s(t)} \frac{\ddd}{\ddd t} s(t),
\end{equation}
where
\begin{equation}
\label{eq:survival}
s(t) = \mathbb{P}(\tau>t)
\end{equation}
is the probability to survive up to time $t$. It is easily checked that Gompertz law implies
\[
s(t) = C\, \eee^{-\frac{\alpha}{\beta} \eee^{\beta t}},
\]
with $C$ an integration constant. In particular, the density of the lifetime distribution equals
\[
\frac{\ddd}{\ddd t}(1-s(t)) = C\, \alpha \, \eee^{\beta t - \frac{\alpha}{\beta} \eee^{ \beta t}}.
\]

Given that ${\bf Z}(t) = z \in S_{\rm{predeath}}$ with $(z_1,z_2)=(1,0)$, the probability that within time $0 < \Delta \ll 1$ the transition to $S_{\rm{death}}$ is made before any of the aging nodes changes equals $F(z)\Delta + o(\Delta)$ with
\[
F(z) = \frac{A_+(\hat{f}_2(z))}{\sum_{i=1}^n 1_{z_i=0}\, A_+(\hat{f}_i(z)) + \sum_{i=1}^n 1_{z_i=1}\, A_-(\hat{f}_i(z))},
\]
where
\begin{equation*}
\hat{f}_i (z) = \frac{1}{|N(i)|} \sum_{j \in N (i)} 1_{z_j = 1}.
\end{equation*}
A similar formula holds when $(z_1,z_2) = (0,1)$ with in the numerator $A_+(\hat{f}_1(z))$. The mortality rate at time $t$ therefore equals
\[
m(t) = \mathbb{E}\big [F({\bf Z}(t))\,1_{{\bf Z}(t) \in S_{\rm{predeath}}} \mid \tau > t \big] .
\]
Since this quantity depends on the structure of the health network and the parameters $r_+,r_-,\Gamma_0,R$, it is a challenging task to find out how it depends on $t$.


\subsection{Summary}

The above model has two key ingredients:
\begin{itemize}
\item \emph{Choice of network structure:}
\begin{itemize}
\item Aging graph $G$, scale-free and disassortative.
\item Number of nodes $n$, degrees of the nodes $N(i)$, $i = 1,\ldots,n$.
\end{itemize}
\item \emph{Choice of network dynamics:}
\begin{itemize}
\item Aging dynamics ${\bf Z} = ({\bf Z}(t))_{t \geq 0}$, time-inhomogeneous Markovian.
\item Overall rate $\Gamma_0$ and adjustment factor $R$.
\item Damage and recovery rates $r_+$, $r_-$.
\end{itemize}
\end{itemize}


\section{Mathematical analysis}
\label{sec:analysis}


\subsection{Underlying network and Poissonisation}
\label{sec:netmot}

Below $G$ lies a more complex and unobservable network, controlling the health of the individual. We view $G$ as the observable resultant of this network at the level of health indicators, and ${\bf Z}$ as the resultant of the dynamics through this network. In Figure~\ref{fig:network_expanded} the additional nodes of the more complex and unobservable network are depicted in grey (representing the `microscopic level'). These additional nodes act in clusters on each single black node (representing the `mesoscopic level').

\begin{figure}[htbp]
\centering
\vspace{-0.2cm}
\includegraphics[width=0.4\textwidth]{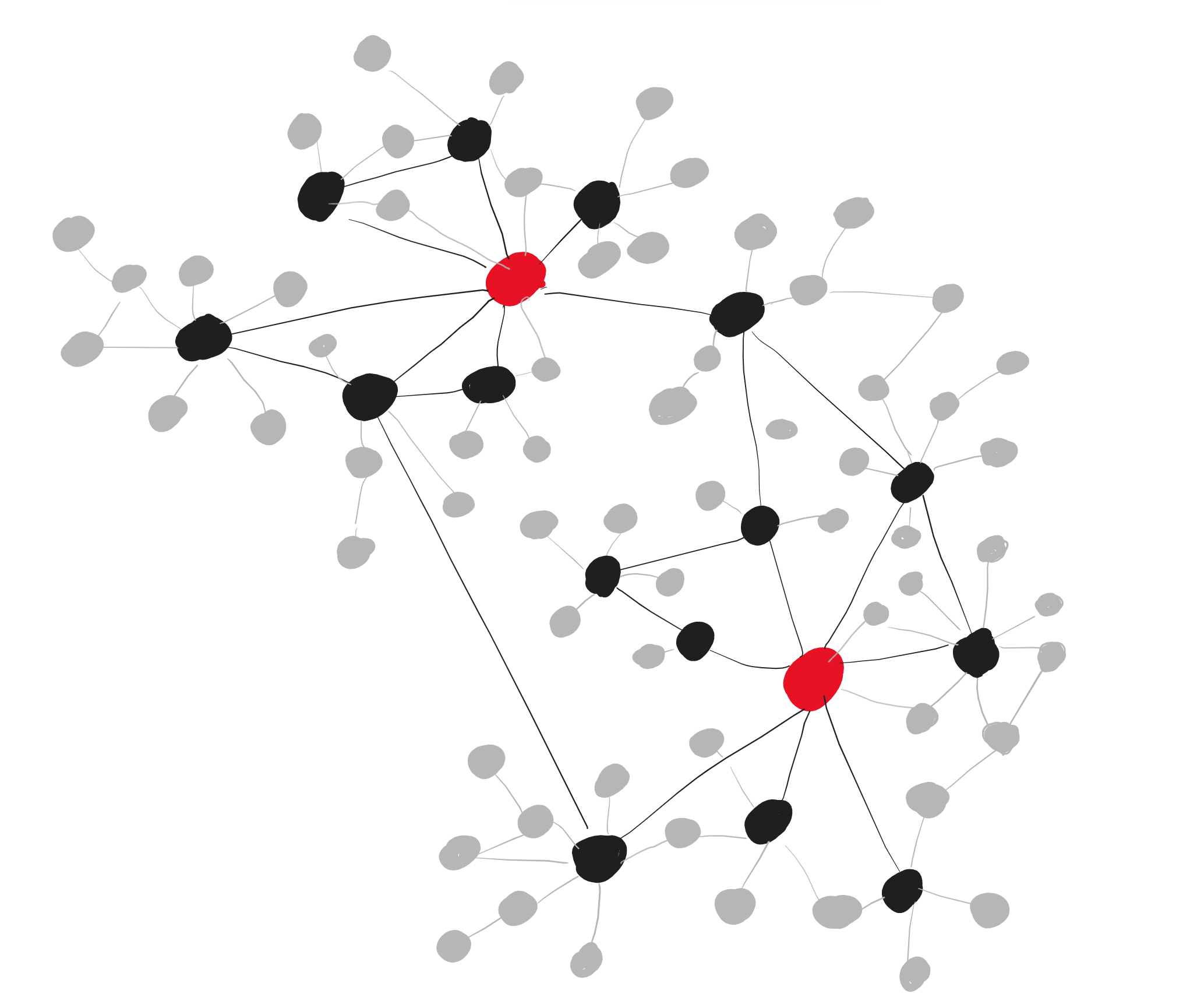}
\caption{\small An example of an extended health network (compare with Figure~\ref{fig:network}).}
\label{fig:network_expanded}
\end{figure}

We next argue why this choice of model is \emph{not ad hoc}, but rather provides a \emph{causal} approach to mortality, capturing both robustness and universality. In particular, we provide an explanation for why deficits and repairs accumulate at \emph{exponential rates} as in \eqref{expdamheal}. We show that this situation arises naturally when we interpret the state of an aging node or a mortality node (the observable black and red nodes) as representing the accumulated state of a large number of subnodes (the unobservable grey nodes). In other words, each node sees a \emph{superposition} of processes at the subnodes.

The following theorem, taken from \cite[Theorem (3.10)]{C1972}, says that \emph{any} superposition of a large number of \emph{sparse} point processes on $\mathbb{R}$ is close to a Poisson point process on $\mathbb{R}$. It does not need \emph{any} assumptions on the distribution of the constituent point process other than that these are \emph{thin} and \emph{plenty}.

\begin{Theorem}{\bf [Poisson]}
\label{thm:Poisson}
For $n \in \mathbb{N}$, let $(M_k)_{k=1}^n$ be a sequence of independent point processes on $\mathbb{R}$ such that, for some measure $\mu$ on $\mathbb{R}$ and all finite intervals $I \subset \mathbb{R}$,
\begin{itemize}
\item[(1)]
$\lim_{n\to\infty} \sum_{k=1}^n \mathbb{P}(M_k(I) = 1) = \mu(I)$,
\item[(2)]
$\lim_{n\to\infty} \sup_{1 \leq k \leq n} \mathbb{P}(M_k(I) \geq 1) = 0$,
\item[(3)]
$\lim_{n\to\infty} \sum_{k=1}^n \mathbb{P}(M_k(I) \geq 2) = 0$,
\end{itemize}
where $M_k(I)$ denotes the number of points that $M_k$ puts inside $I$. Then $M^n = \cup_{k=1}^n M_k$ converges weakly as $n\to\infty$ to the Poisson point process on $\mathbb{R}$ with intensity measure $\mu$.
\end{Theorem}

\noindent
Condition (1) guarantees that the superposition $M^n$ has an intensity measure $\mu_n$ that converges to a limiting intensity measure $\mu$. The interest is in the situation where $\mu$ is \emph{non-trivial}, i.e., is neither zero nor infinity. Condition (2) says that the superposition is \emph{infinitesimal}, i.e., each constituent point process makes a vanishing contribution. Condition (3) says that the superposition is \emph{uniformly sparse}, i.e., points do not cluster. There is no (!) condition on the nature of the constituent point processes, which may have strong dependencies and may be far from being Poisson themselves. It is only through the \emph{superposition} that the Poisson nature comes out in the limit, subject to the three conditions. For the special case where, for each $1 \leq k \leq n$, $M_k$ is a \emph{stationary renewal process} with \emph{mean interarrival time} $1/\lambda_k \in (0,\infty)$, condition (1) reads $\lim_{n\to\infty} \sum_{k=1}^n \lambda_k = \lambda \in (0,\infty)$, and the limit is the Poisson point process with intensity $\lambda$, i.e., exponentially distributed interarrival times with mean $1/\lambda$.

We use Theorem~\ref{thm:Poisson} with the role of space $\mathbb{R}$ taken over by time $[0,\infty)$. In that context the theorem says that the superposition of a large number of possibly time-dependent and mutually dependent clocks that ring rarely behaves like a single Poisson clock that rings at a possibly time-dependent rate.


\subsection{Mean-field and homogeneity assumptions}

Now that we have set up the model, one way to proceed would be to make simple choices for $G$, like a \emph{complete graph} with all the nodes connected, or a \emph{bipartite graph} with the mortality nodes and the aging nodes forming two communities in the network, or a \emph{scale-free graph} with the mortality nodes forming the hubs in the network and the aging nodes not being hubs. This is the approached followed by \cite{MRFR2017} and the related literature discussed in the Introduction. Here we follow a new, analytical approach, which requires only very generic structural assumptions on the network. Results are derived by approximate computations based on a \emph{mean-field assumption} and a \emph{homogeneity assumption}. These assumptions are plausible because of the presumed hub structure of the mortality nodes. Below, whenever we write $\approx$ we refer to an approximation implied either by the two assumptions or by some other simplification. It remains a \emph{mathematical challenge} to provide error bounds.


We approximate the dependence of a node on the fraction of damaged neighbours of that node, as expressed in \eqref{eq:frailty}, by the average of the nodes that are in state $1$. To that end, we put
\begin{equation}
\label{eq:frailtyl}
\hat{p}(t) = \frac{1}{n}  \sum_{j=1}^n 1_{Z_j(t) = 1}, \qquad p(t) = \mathbb{E}[\hat p(t)]
\end{equation}
and we refer to $ p ( t ) $ as \emph{damage fraction} at time $t$.
For $n$ large, the $Z_j(t)$ for nodes that are far apart are nearly independent, so that the average over the nodes acts like a \emph{law of large numbers}.

\medskip\noindent
$\blacktriangleright$ We make the following \emph{mean-field assumption} (compare with \eqref{expdamheal}--\eqref{eq:Apmdef}):

\medskip\noindent
{\sc Assumption (A1).}
{\rm The rates at time $t$ for node $i$ to go from healthy to damaged, respectively, from damaged to healthy are given by
\[
\Gamma_+ (i,t) \approx A_+(p(t)), \qquad \Gamma_- (i,t) \approx A_-(p(t)),
\]
with $p(t)$ the average fraction of damaged nodes at time $t$. Note that the approximate rates do not depend on $i$.} \hfill$\spadesuit$

\medskip\noindent
$\bullet$ We comment on {\sc Assumption (A1)}. If the number of nodes $n$ is large, and the neighbourhood $N(i)$ of node $i$ is large as well, then we can approximate (recall \eqref{eq:frailty})
\begin{equation}
\label{eq:conc}
f_i(t) = \frac{1}{|N(i)|} \sum_{j \in N (i)} 1_{Z_j(t) = 1} \approx \hat{p}(t) \approx p(t), \qquad  i = 1,2.
\end{equation}
Indeed, in a \emph{scale-free} and \emph{disassortative} network there are many nodes with high degree, while the neighbourhoods of many nodes are large as well yet cover only a small fraction of the nodes. The former ensures that a law of large numbers is in force, while the latter ensures that the nodes are more or less independent of each other. Combining \eqref{expdamheal} with \eqref{eq:conc}, we approximate
\begin{equation}
\label{eq:exp}
\begin{aligned}
\Gamma_+(i,t) &= A_+(f_i(t)) \approx \mathbb{E}[A_+(\hat{p}(t))] = \mathbb{E}\left[\Gamma_0\,\eee^{r_+ \hat{p}(t)}\right]
\approx \Gamma_0\,\eee^{r_+\mathbb{E}[\hat{p}(t)]} = A_+(p(t)),\\
\Gamma_-(i,t) &= A_-(f_i(t)) \approx \mathbb{E}[A_-(\hat{p}(t))] = \mathbb{E}\left[\frac{\Gamma_0}{R}\,\eee^{-r_- \hat{p}(t)}\right]
\approx \frac{\Gamma_0}{R}\,\eee^{-r_-\mathbb{E}[\hat{p}(t)]} = A_-(p(t)),
\end{aligned}
\end{equation}
where in the second approximation we bring the expectation to the exponent (later we return to this simplification).

\begin{Remark}
{\rm The health network introduced in \cite{MRFR2017} also includes so-called \emph{frailty nodes}: a small number of aging nodes that are highly connected. These frailty nodes allow for the modelling of the so-called \emph{frailty index} \cite{searle2008standard}, which is of interest in health sciences. However, due to the mean-field approximation put forward in \eqref{eq:conc}, according to which all nodes are treated in the same way, our approach does not allow for a closer analysis of the frailty nodes and the frailty index.}
\end{Remark}

\begin{Remark}
\label{rem:Jensen}
{\rm For short times, when few nodes are damaged, the law of large numbers that underlies the first approximation in \eqref{eq:exp} is less sharp than for large times, when many nodes are damaged. In the second approximation in \eqref{eq:exp} we brought the expectation to the exponent. However, by Jensen's inequality for expectations of convex functions, we actually have
\[
\mathbb{E}\left[\eee^{r_+ \hat{p}(t)}\right] \geq \eee^{r_+\mathbb{E}[\hat{p}(t)]},
\qquad \mathbb{E}\left[\eee^{-r_- \hat{p}(t)}\right] \geq \eee^{-r_-\mathbb{E}[\hat{p}(t)]},
\]
i.e., the second approximation \emph{undershoots} the two rates, in particular, the rate for a node to become damaged.}
\end{Remark}

\medskip\noindent
$\blacktriangleright$ We make the following \emph{homogeneity assumption}:

\medskip\noindent
{\sc Assumption (A2).}
{\rm The states in $S_{\rm{other}}$, $S_{\rm{predeath}}$, $S_{\rm{death}}$ are aggregated into three single states. The transition rates between these single states at time $t$ are as in Figure~\ref{fig:aggregated}.}\hfill$\spadesuit$

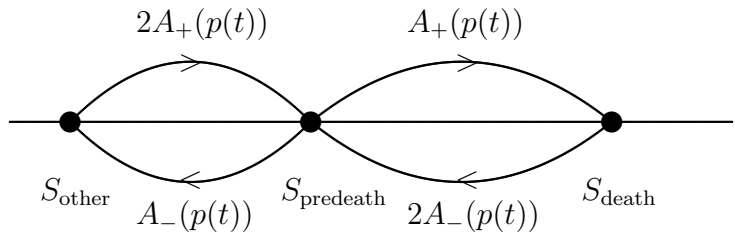
\begin{figure}[htbp]
\begin{center}
\setlength{\unitlength}{0.8cm}
\begin{picture}(10,4)(0,-2)
{\thicklines
\qbezier(0,0)(6,0)(12,0)
\qbezier(1,0)(3,2)(5,0)
\qbezier(5,0)(3,-2)(1,0)
\qbezier(5,0)(7.5,2)(10,0)
\qbezier(10,0)(7.5,-2)(5,0)
}
\put(0.5,-1.3){$S_{\rm {other}}$}
\put(4.5,-1.3){$S_{\rm {predeath}}$}
\put(9.5,-1.3){$S_{\rm {death}}$}

\put(1,0){\circle*{0.35}}
\put(5,0){\circle*{0.35}}
\put(10,0){\circle*{0.35}}

\put(2.1,1.5){$2A_+(p(t))$}
\put(2.1,-1.7){$A_-(p(t))$}
\put(6.6,1.5){$A_+(p(t))$}
\put(6.6,-1.7){$2A_-(p(t))$}

\put(2.8,-1.1){$<$}
\put(2.8,0.85){$>$}
\put(7.4,0.85){$>$}
\put(7.4,-1.1){$<$}

\end{picture}
\end{center}
\caption{\small The aggregated time-inhomogeneous Markov chain.}
\label{fig:aggregated}
\end{figure}

\medskip\noindent
$\bullet$ We comment on {\sc Assumption (A2)}. The transition from $S_{\rm{other}}$ to $S_{\rm{predeath}}$ occurs when both two mortality nodes are healthy and one of them switches to damaged. The transition from $S_{\rm{predeath}}$ to $S_{\rm{death}}$ occurs when one mortality node is damaged, the other mortality node is healthy and switches to damaged. For the reverse transitions a similar argument applies.

\medskip\noindent
$\blacktriangleright$
Combining the mean-field assumption and the homogeneity assumption, we obtain the following.

\begin{Lemma}
\label{lem:a3}
Under {\sc Assumptions (A1)--(A2)}, $p(t)$ is the solution of the autonomous differential equation
\begin{equation}
\label{eq:pdiff}
\frac{\ddd}{\ddd t}\, p(t) = (1-p(t))\,A_+(p(t)) - p(t)\,A_-(p(t)), \qquad p(0) = 0.
\end{equation}
\end{Lemma}

\begin{proof}
The rate at which one of the $n(1-p(t))$ healthy nodes becomes damaged equals $n$ times the first term, while the rate at which one of the $np(t)$ damaged nodes becomes healthy equals $n$ times the second term. The difference is the net rate at which the number of damaged nodes increases or decreases. Divide by $n$ to get the fraction of damaged nodes. Note that $p(0)=0$ because initially all nodes are healthy.
\end{proof}


\subsection{The approximate mortality rate}

Our task is to derive the Gompertz law based on {\sc Assumption (A1)} and {\sc Assumption (A2)}. In this section we derive a formula for the mortality rate. The argument proceeds in steps, listed as Lemmas~\ref{lem:mt}--\ref{lem:a2} below, leading up to Theorem~\ref{th:Glaw}. In Section~\ref{sec:anal} we analyse this formula for several choices of the model parameters to get a feel for how the Gompertz law may emerge. We also derive a formula for the damage fraction at death, stated in Theorem~\ref{th:frdeath} below.

\begin{Lemma}
\label{lem:mt}
For every $t \geq 0$,
\begin{equation}
\label{eq:mortalitysplit}
\begin{aligned}
m(t) &= \lim_{\Delta \downarrow 0} \frac{1}{\Delta}\,\mathbb{P}(\tau \leq t + \Delta \mid \tau \geq t)\\
&= \lim_{\Delta \downarrow 0} \frac{1}{\Delta}\,
\mathbb{P}\big(\exists\,0 \leq u \leq \Delta\colon\, {\bf Z}(t+u) \in S_{\rm{death}} \mid {\bf Z}(t) \in S_{\rm{predeath}}\big)\\
&\qquad \times \mathbb{P}\big({\bf Z}(t) \in S_{\rm{predeath}} \mid {\bf Z}(u) \not \in S_{\rm{death}}, 0 \leq u \leq t \big).
\end{aligned}
\end{equation}
\end{Lemma}

\begin{proof}
Abbreviate
\[
\begin{aligned}
A_t &= \{{\bf Z}(u) \not \in S_{\rm{death}}, 0 \leq u \leq t\},\\
B_t &= \{{\bf Z}(t) \in S_{\rm{predeath}}\},\\
C_{t,\Delta} &=  \{\exists\,0 \leq u \leq \Delta\colon\, {\bf Z}(t+u) \in S_{\rm{death}}\}.
\end{aligned}
\]
Over short time intervals, ${\bf Z}$ can only reach $S_{\rm{death}}$ by going from $S_{\rm{predeath}}$ to $ S_{\rm{death}}$ in a single jump. Hence, for $\Delta \downarrow 0$,
\[
\{t \leq  \tau \leq t + \Delta\} \sim A_t \cap B_t \cap C_{t,\Delta},
\]
where $\sim$ means equality up to an error that vanishes with $\Delta \downarrow 0$. Since
\[
\{\tau \geq t\} = A_t,
\]
it follows that, for $\Delta \downarrow 0$,
\[
\begin{aligned}
\mathbb{P}(\tau \leq t+\Delta \mid  \tau \geq t)
&\sim \frac{\mathbb{P}(A_t \cap B_t \cap C_{t,\Delta} )}{\mathbb{P}(A_t)}
= \frac{\mathbb{P}(A_t \cap C_{t,\Delta} \mid B_t)\,\mathbb{P}(B_t)}{\mathbb{P}(A_t)}\\
&= \frac{\mathbb{P}(C_{t,\Delta} \mid B_t)\,\mathbb{P}(A_t \mid B_t)\,\mathbb{P}(B_t)}{\mathbb{P}(A_t)}
= \mathbb{P}(C_{t,\Delta} \mid B_t)\,\mathbb{P}(B_t \mid A_t),
\end{aligned}
\]
where the second equality uses the Markov property at time $t$.
\end{proof}

From the mean-field assumption we obtain the following representation of the first factor in \eqref{eq:mortalitysplit}.

\begin{Lemma}
\label{lem:a1}
Under {\sc Assumption (A1)},
\[
\lim_{\Delta \downarrow 0} \frac{1}{\Delta}
\mathbb{P}\big(\exists\,0 \leq u \leq \Delta\colon\, {\bf Z}( u+t) \in S_{\rm{death}} \mid  {\bf Z}(t) \in S_{\rm{predeath}}\big)
 \approx A_+(p(t)).
\]
\end{Lemma}

\begin{proof}
The rate to jump from a predeath state to the death state is approximately $A_+(p(t))$, because one mortality node is damaged, the other mortality node is healthy, and the latter needs to switch to damaged.
\end{proof}

From the homogeneity assumption we obtain the following representation of the second factor in \eqref{eq:mortalitysplit}.

\begin{Lemma}
\label{lem:a2}
Under {\sc Assumption (A2)},
\[
\mathbb{P}\big({\bf Z}(t) \in S_{\rm{predeath}} \mid {\bf Z}(u) \not \in S_{\rm{death}}, 0 \leq u \leq t \big)
\approx \chi(t)
\]
with $\chi(t)$ the solution of the differential equation
\begin{equation}
\label{eq:chidiff}
\frac{\ddd}{\ddd t}\,\chi(t) = (1-\chi(t))\,2A_+(p(t)) - \chi(t)\,A_-(p(t)), \qquad \chi(0) = 0.
\end{equation}
\end{Lemma}

\begin{proof}
The law of the time-homogeneous Markov process ${\bf Z}$ on the set of three states $\{S_{\rm{other}},S_{\rm{predeath}},S_{\rm{death}}\}$ \emph{conditional} on not entering $S_{\rm{death}}$ is the same as the law of the time-homogeneous Markov process ${\bf Z}'$ on the set of two states $\{S_{\rm{other}},S_{\rm{predeath}}\}$ where the rates to and from $S_{\rm{death}}$ are \emph{set to zero} (see Figure~\ref{fig:aggregated}). Hence $\chi(t) = \mathbb{P}\big({\bf Z}'(t) \in S_{\rm{predeath}})$. This probability solves \eqref{eq:chidiff}, because the net rate at which $\chi(t)$ increases is the difference of the rates at which ${\bf Z}'$ jumps into, respectively, jumps out of the state $S_{\rm{predeath}}$. Note that $\chi(0)=0$ because initially all nodes are healthy.
\end{proof}

\medskip
Combining Lemmas~\ref{lem:a3}--\ref{lem:a2}, we can identify the approximate mortality rate.

\begin{Theorem}{\bf [Approximate mortality rate]}
\label{th:Glaw}
Under {\sc Assumptions (A1)--(A2)},
\begin{equation}
\label{eq:GMlaw}
m(t) \approx A_+(p(t))\,\chi(t).
\end{equation}
\end{Theorem}

\noindent
The solution of \eqref{eq:pdiff} tells us how $p(t)$ grows as a function of $t$. Once this profile has been determined, the solution of $\eqref{eq:chidiff}$ tells us how $\chi(t)$ grows as a function of $t$. The two together fix how $m(t)$ grows as a function of $t$.

A further interesting quantity is the \emph{damage fraction at death}.

\begin{Theorem}{\bf [Damage fraction at death]}
\label{th:frdeath}
\begin{equation}
\label{eq:frdeath}
\mathbb{E}[\hat{p}(\tau)] \approx \int_0^\infty \ddd t\, A_+(p(t))\,p(t)\,\chi(t)\,\exp\left[-\int_0^t \ddd u\, A_+(p(u))\,\chi(u)\right].
\end{equation}
\end{Theorem}

\begin{proof}
By \eqref{eq:deathtime}, \eqref{eq:survival} and \eqref{eq:frailtyl},
\[
\mathbb{E}[\hat{p}(\tau)] \approx \int_0^\infty p(t)\,\mathbb{P}(\tau \in \ddd t) = \int_0^\infty p(t) \left[-\frac{\ddd}{\ddd t} s(t)\right]\,\ddd t.
\]
It follows from \eqref{eq:mortality} that
\[
s(t) = \exp\left[-\int_0^t m(u)\,\ddd u\right]
\]
and hence
\[
-\frac{\ddd}{\ddd t} s(t) = m(t)\,\exp\left[-\int_0^t m(u)\,\ddd u\right].
\]
Use \eqref{eq:GMlaw} to get the claim.
\end{proof}


\subsection{Analysis of the approximate mortality rate}
\label{sec:anal}

Having derived an approximate formula for the mortality rate $m(t)$ in terms of the differential equations in \eqref{eq:pdiff} and \eqref{eq:chidiff}, our next task is to see whether this formula produces the Gompertz law for times that are neither too small nor too large.

We first \emph{scale time} by $1/\Gamma_0$ so as to remove the parameter $\Gamma_0$. Putting
\begin{equation}
\label{eq:micmac}
p^*(s) = p\left(\frac{s}{\Gamma_0}\right), \qquad \chi^*(s) = \chi\left(\frac{s}{\Gamma_0}\right),
\qquad m^*(s) = m\left(\frac{s}{\Gamma_0}\right),
\end{equation}
and
\begin{equation}
\label{eq:A*defs}
A^*_+(p) = \eee^{r_+ p}, \qquad A^*_-(p) = \frac{1}{R}\,\eee^{-r_- p}, \qquad p \in [0,1],
\end{equation}
we can rewrite \eqref{eq:GMlaw} as
\begin{equation}
\label{eq:m*def}
m^*(s) \approx \Gamma_0\,A^*_+(p^*(s))\,\chi^*(s),
\end{equation}
where $p^*(s)$ and $\chi^*(s)$ solve the differential equations (recall \eqref{eq:pdiff} and \eqref{eq:chidiff})
\begin{equation}
\label{eq:scaldiff}
\begin{array}{lll}
&\frac{\ddd}{\ddd s}\, p^*(s) = (1-p^*(s))\,A^*_+(p^*(s)) - p^*(s)\,A^*_-(p^*(s)), &p^*(0) = 0,\\[0.3cm]
&\frac{\ddd}{\ddd s}\, \chi^*(s) = (1-\chi^*(s))\,2A^*_+(p^*(s)) - \chi^*(s)\,A^*_-(p^*(s)), &\chi^*(0) = 0.
\end{array}
\end{equation}
Note that $\Gamma_0$ drops out of \eqref{eq:scaldiff} because
\begin{equation}
\label{eq:tsrel}
t = \frac{s}{\Gamma_0} \quad \longrightarrow \quad \frac{\ddd}{\ddd s} = \Gamma_0 \frac{\ddd}{\ddd t},
\end{equation}
so that a factor $\Gamma_0$ can be cancelled on both sides of the differential equations. Think of $t$ as the `microscopic' time scale on which the single nodes in the health network evolve, and of $s$ as the `macroscopic' time scale on which the network consisting of many nodes evolves as a whole. If $\Gamma_0 = \frac{1}{C}/{\rm year}$, then one unit of macroscopic time corresponds to $C$ years of human aging (think of $C$ as a `calibration parameter'.)

Since $p^*(s)$ and $\chi^*(s)$ are differentiable in $s$, they are continuous in $s$ as well. Note that $p^*(s)$ is strictly increasing in $s$ with $\lim_{s\to\infty} p^*(s) = p_\dagger$, where $p_\dagger$ solves the equation $(1-p_\dagger)\,A^*_+(p_\dagger) = p_\dagger\,A^*_-(p_\dagger)$, so
\begin{equation}
\label{eq:pdagger}
p_\dagger = \frac{A^*_+(p_\dagger)}{A^*_+(p_\dagger)+A^*_-(p_\dagger)} = \frac{1}{1 + \frac{1}{R}\, \eee^{-(r_+ + r_-)p_\dagger}}.
\end{equation}
Note that also $\chi^*(s)$ is strictly increasing in $s$ with $\lim_{s\to\infty} \chi^*(s) = \chi_\dagger$, where $\chi_\dagger$ solves the equation $(1-\chi_\dagger)\,2A^*_+(p_\dagger) = \chi_\dagger\, A^*_-(p_\dagger)$, so
\begin{equation}
\label{eq:chidagger}
\chi_\dagger = \frac{2A^*_+(p_\dagger)}{2A^*_+(p_\dagger)+A^*_-(p_\dagger)} = \frac{2p_\dagger}{1+p_\dagger}.
\end{equation}
It is further evident that $\frac{\ddd}{\ddd s} p^*(0) = 1$ and $\frac{\ddd}{\ddd s} \chi^*(0) = 2$. For later use we need the following observation.

\begin{Lemma}
\label{lem:sandwich}
$p^*(s) <  \chi^*(s) < 2p^*(s)$ for all $s > 0$.
\end{Lemma}

\begin{proof}
Put
\begin{equation*}
\Delta_1(s) = \chi^*(s) - p^*(s), \qquad \Delta_2(s) = 2p^*(s) - \chi^*(s).
\end{equation*}
It follows from \eqref{eq:scaldiff} that
\begin{equation*}
\begin{aligned}
\frac{\ddd}{\ddd s} \Delta_1(s) &= - \big[A^*_+(p^*(s)) + A^*_-(p^*(s))\big]\, \Delta_1(s) + A^*_+(p^*(s))\,(1-\chi^*(s)),\\
\frac{\ddd}{\ddd s} \Delta_2(s) &= - A^*_-(p^*(s))\, \Delta_2(s) + 2A^*_+(p^*(s))\,\Delta_1(s).
\end{aligned}
\end{equation*}
We know that $\Delta_1(0)=0$, $\Delta_1(s)>0$ for small $s>0$, and $\Delta_1(s)$ is continuous and differentiable in $s$. It is impossible  that $\Delta_1(s)$ hits the value $0$ at some $s_1>0$, because this would imply that $\frac{\ddd}{\ddd s} \Delta_1(s_1) = A^*_+(p^*(s_1))\,(1-\chi^*(s_1)) > 0$. Hence $\Delta_1(s)$ is everywhere strictly positive. Similarly, we know that $\Delta_2(0)=0$, $\Delta_2(s)>0$ for small $s>0$, and $\Delta_2(s)$ is continuous and differentiable in $s$. It is impossible that $\Delta_2(s)$ hits the value $0$ at some $s_2>0$, because this would imply that $\frac{\ddd}{\ddd s} \Delta_2(s_2) = 2A^*_+(p^*(s_2))\,\Delta_1(s_2) > 0$. Hence $\Delta_2(s)$ is everywhere strictly positive as well.
\end{proof}

Figures~\ref{fig:math}--\ref{fig:mathmort} show plots of $p^*(s)$, $\chi^*(s)$ and $\ln [A^*_+(p^*(s))\,\chi^*(s)]$ for $r_+ = r_- \in \{1,2,5\}$ and $R=1$ (carried out with the help of the programming language R). The latter curve is linear only for values of $s$ that are neither too small nor too large. The \emph{range} of $s$-values for which the linear fit is accurate depends on the \emph{choice} of the parameters. The curves for $p^*(s)$, $\chi^*(s)$ tend to be concave for small values of $r_+,r_-$ and convex-concave for large values of $r_+,r_-$. The curve for $\ln [A^*_+(p^*(s))\,\chi^*(s)]$ tends to be concave for small values of $r_+,r_-$ and concave-convex-concave for large values of $r_+,r_-$.

\begin{figure}[htbp]
\centering
\includegraphics[width=0.32\textwidth]{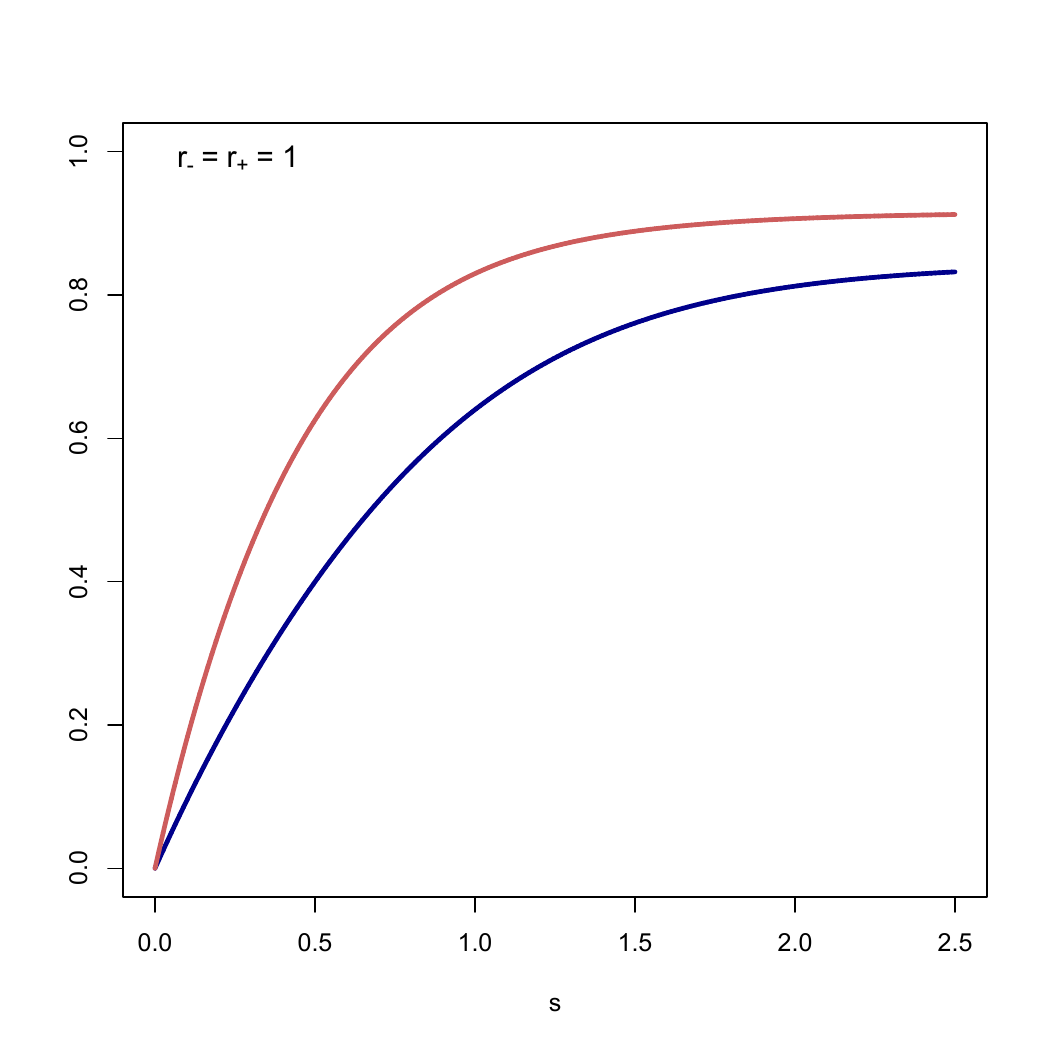}
\includegraphics[width=0.32\textwidth]{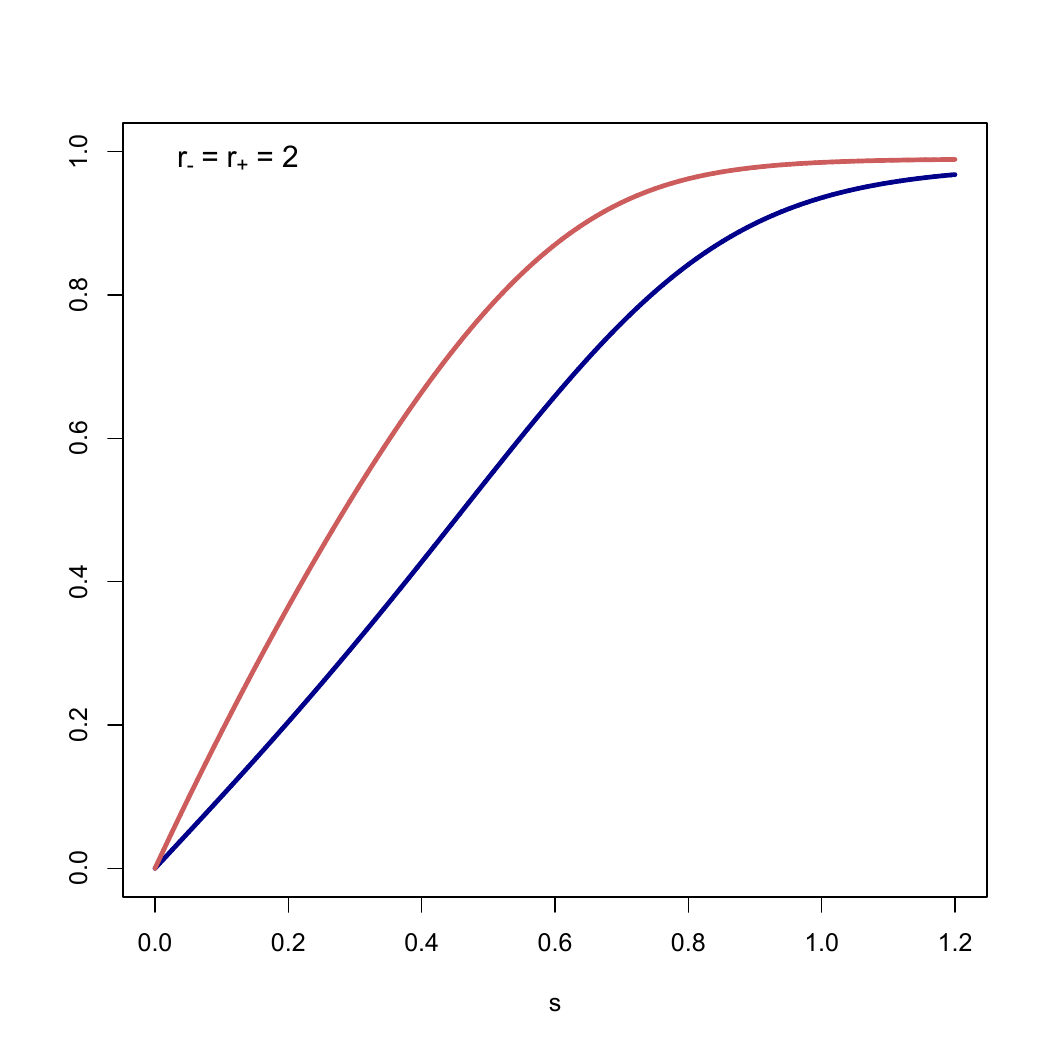}
\includegraphics[width=0.32\textwidth]{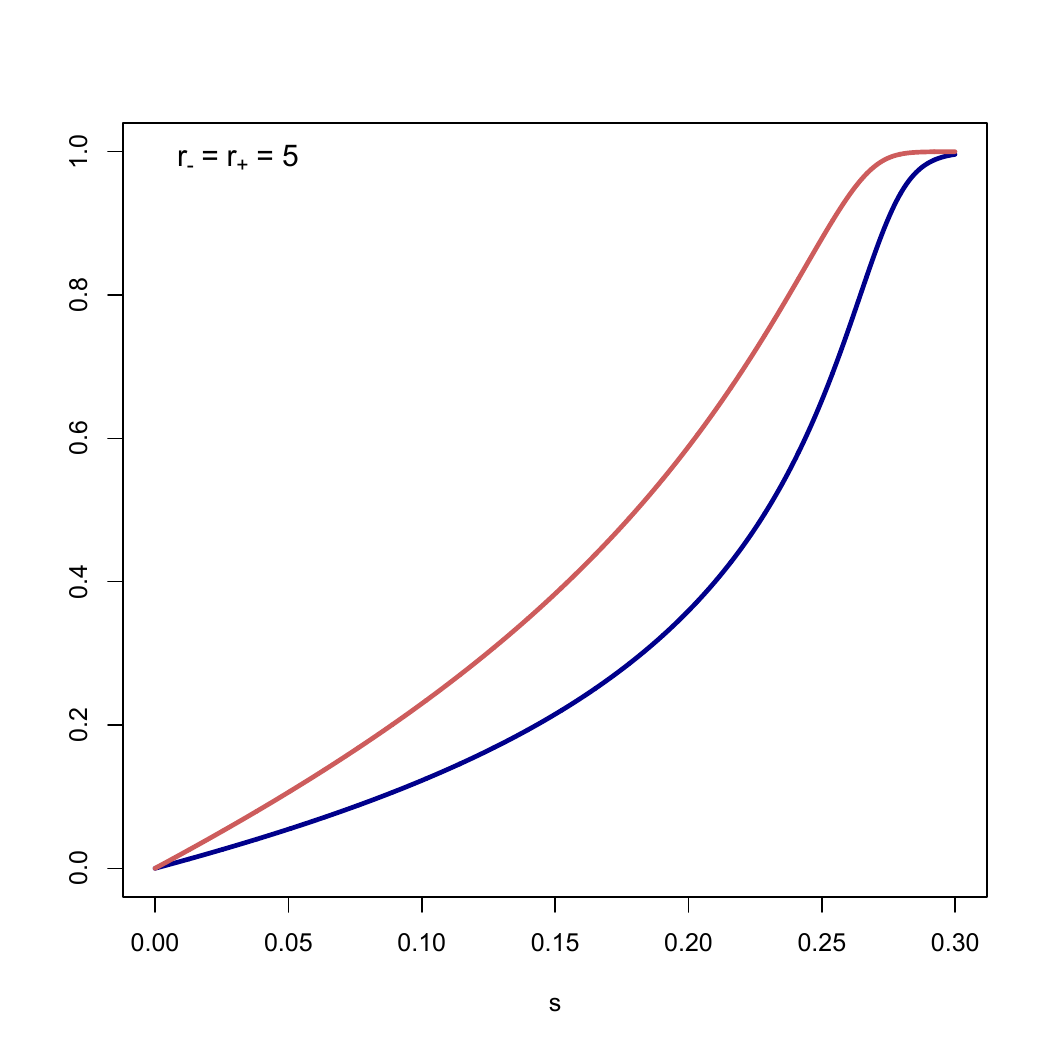}
\caption{\small Plots of $p^*(s)$ (= blue curve) and $\chi^*(s)$ (= red curve) for $r_+ = r_- \in \{1,2,5\}$ and $R=1$.}
\label{fig:math}
\end{figure}

\begin{figure}[htbp]
\centering
\includegraphics[width=0.32\textwidth]{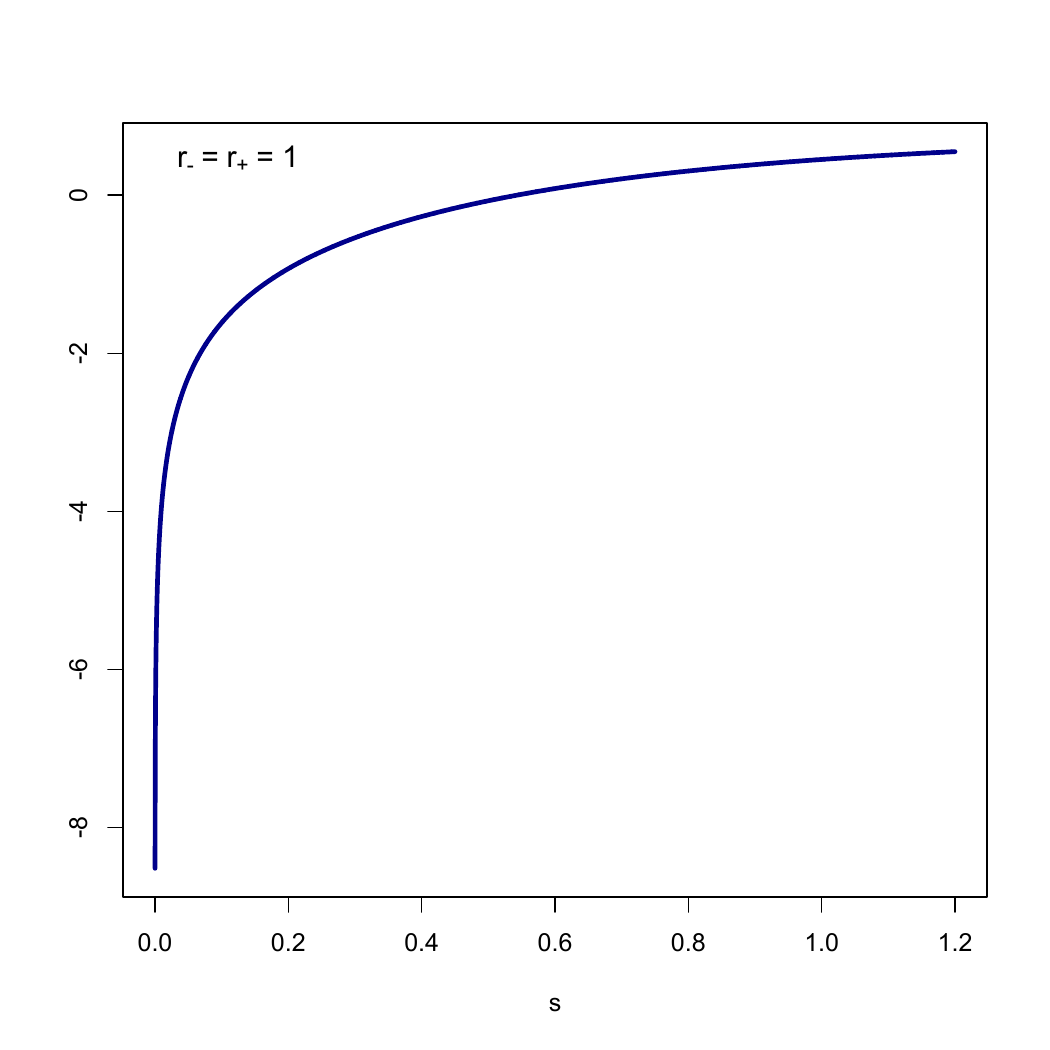}
\includegraphics[width=0.32\textwidth]{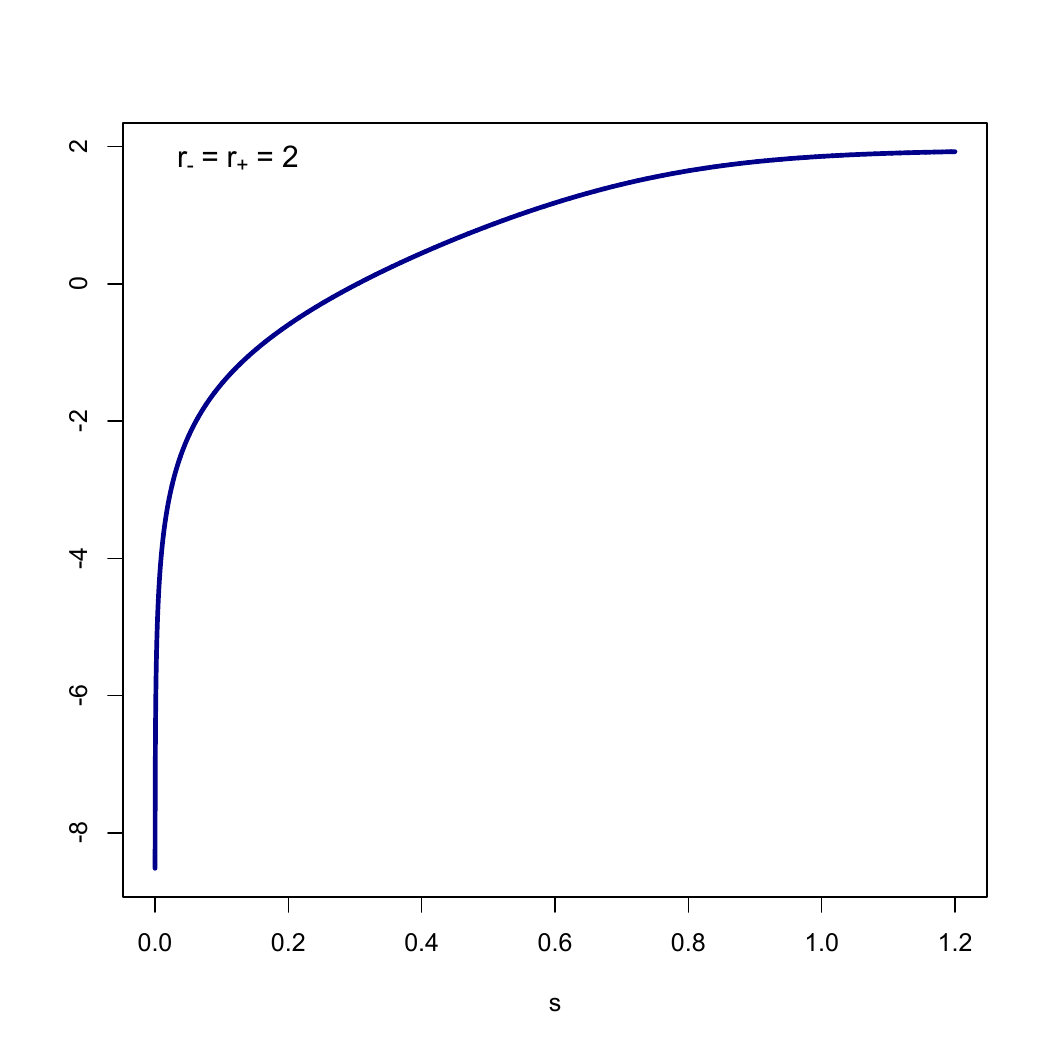}
\includegraphics[width=0.32\textwidth]{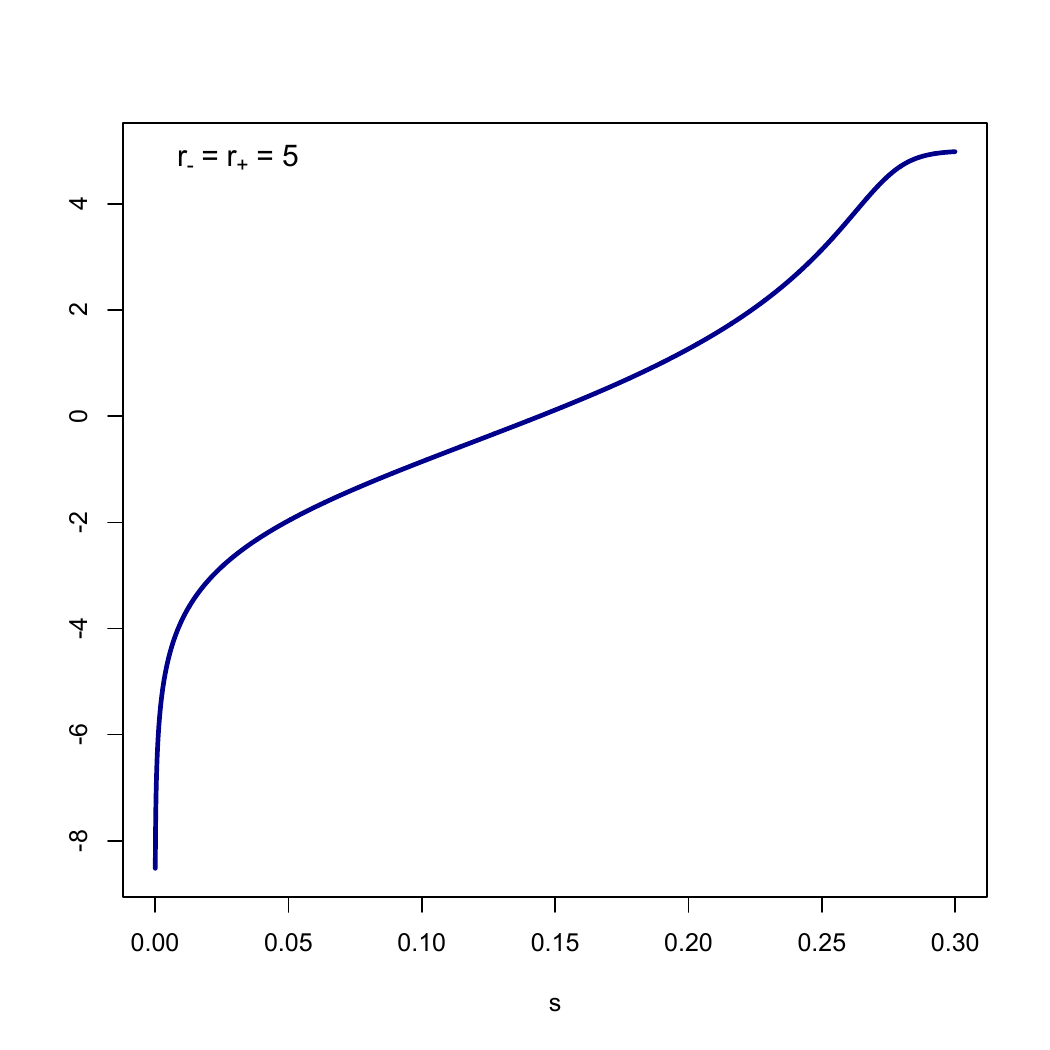}
\caption{\small Plots of $\ln [A^*_+(p^*(s))\,\chi^*(s)]$ for $r_+ = r_- \in \{1,2,5\}$ and $R=1$.}
\label{fig:mathmort}
\end{figure}


\section{Simulations and numerical results}
\label{sec:simulations}

While the aging process ${\bf Z}$ is well-defined for any choice of parameters $n,r_+,r_-,\Gamma_0,R$ with $n\in\mathbb{N}$ and $r_+,r_-,\Gamma_0,R \in (0,\infty)$, it only provides a fair model for the mortality rate when the parameters are chosen properly. In this section we discuss an instance of the aging process provided in \cite{MRFR2017}, where the parameters were chosen to be $n = 10^4$, $r_+ = 10.27$, $r_- = 6.5$, $R = 1.5$ and $\Gamma_0 = 0.00113$/year. For this choice of $\Gamma_0$, one unit of $s$ corresponds to $885$ years, and so $[40,100]$ years corresponds to $s \in [0.045,0.113]$.

Figure~\ref{fig:numplot} plots the damage fraction and the mortality rate in the range between 0 to 100 years of age. For the latter, `observational' means taken from the statistical data for US men in the period 2010-2019 \cite{HMDalt} (recall Figure~\ref{fig:USmen}), while `model' means taken from simulation of the \emph{network dynamics} based on $10^6$ i.i.d.\ samples \cite[Chapter 6]{Fth2020}. It is seen that $\ln m(t)$ is roughly linear over the age interval $[40,100]$ years, which confirms Gompertz law \emph{numerically}. The fitted line is $\ln m(t) \approx -9.76 + 0.085\,t$, which gives
\[
\alpha \approx \eee^{-9.76} \approx 5.8 \times 10^{-5},
\qquad \beta \approx 8.5 \times 10^{-2}.
\]

\begin{figure}[htbp]
\begin{center}
\includegraphics[width=0.45\textwidth]{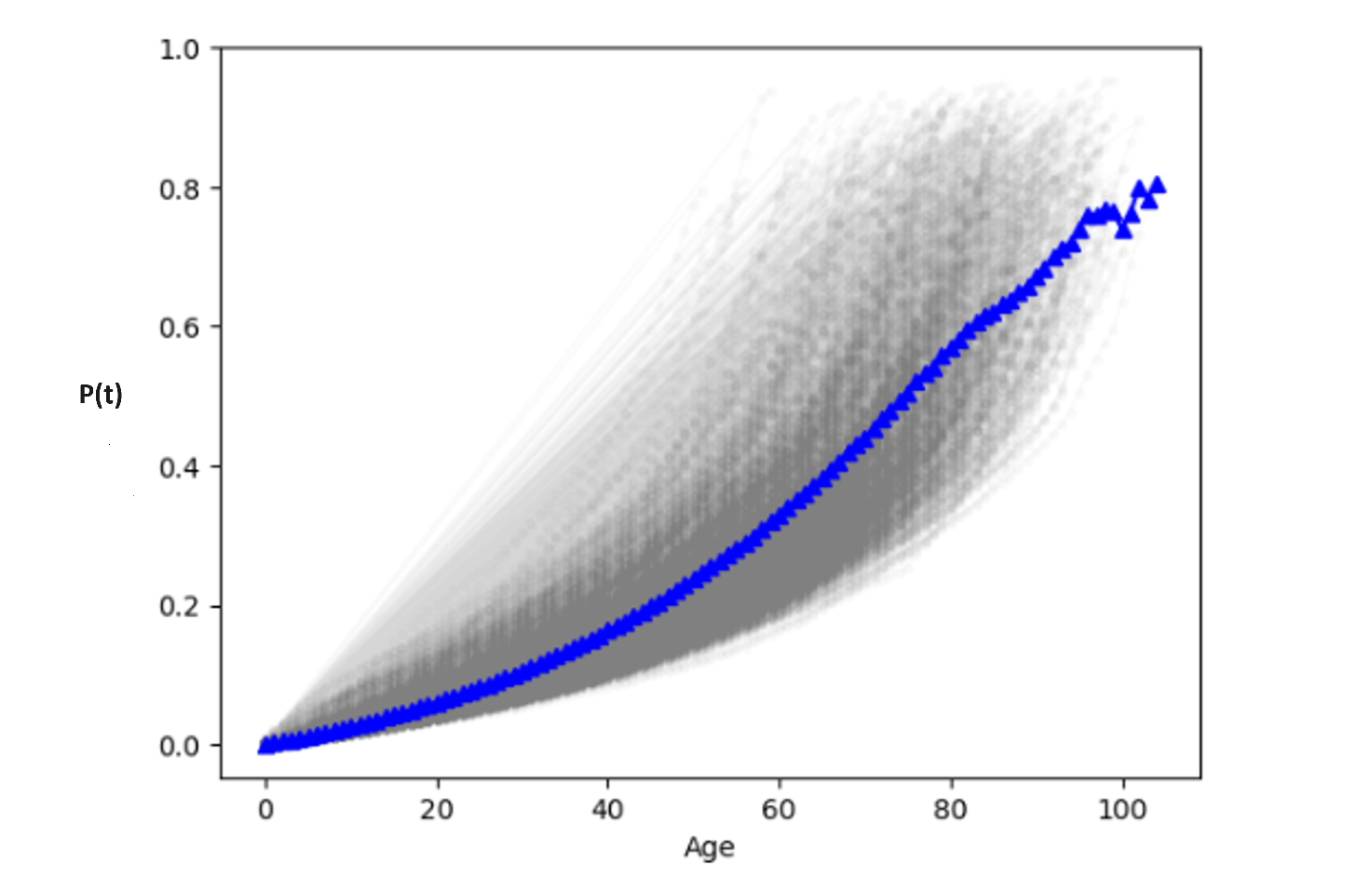}
\includegraphics[width=0.45\textwidth]{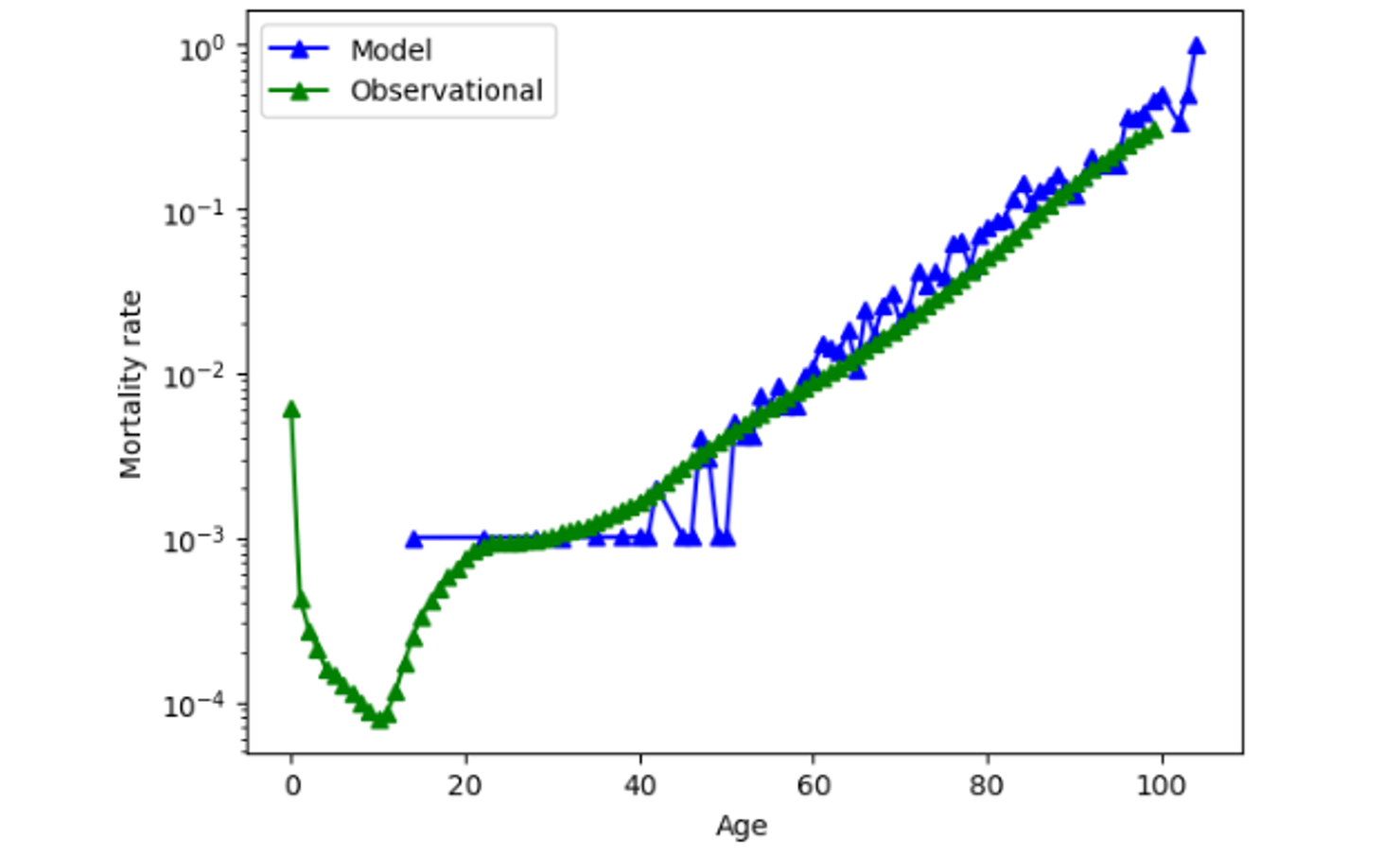}
\end{center}
\caption{\small A simulation of the damage fraction and the mortality rate for the parameter choices in \cite{MRFR2017}. \emph{Left:} The black curves are realisations  of the damage fraction, the blue curve is the average damage fraction. \emph{Right:} The blue curve is the result of simulation of the network dynamics, the green curve represents the statistical data obtained from Figure~\ref{fig:USmen}.}
\label{fig:numplot}
\end{figure}

Figure~\ref{fig:reference-s} shows plots of $p^*(s)$, $\chi^*(s)$ and $\ln[A^*_+(p^*(s))\chi^*(s)]$ based on \eqref{eq:m*def}--\eqref{eq:scaldiff} for the values of $r_+,r_-,R$ used in \cite{MRFR2017}. The latter curve is roughly linear on the interval $s \in [0.02, 0.08]$, but bends down on the left of this interval and bends up on the right.

\begin{figure}[htbp]
\vspace{0.4cm}
\centering
\includegraphics[width=0.4\textwidth]{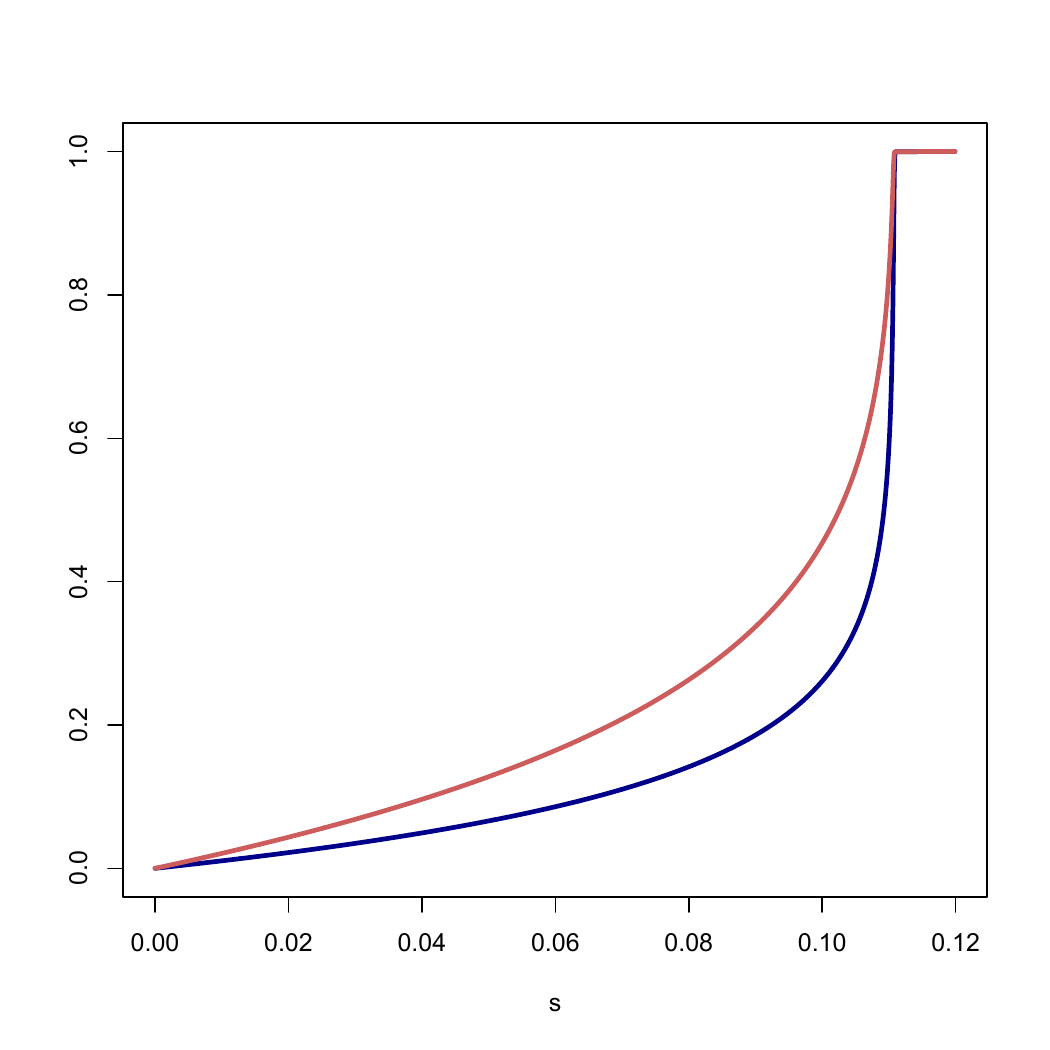} \hspace{1cm}
\includegraphics[width=0.4\textwidth]{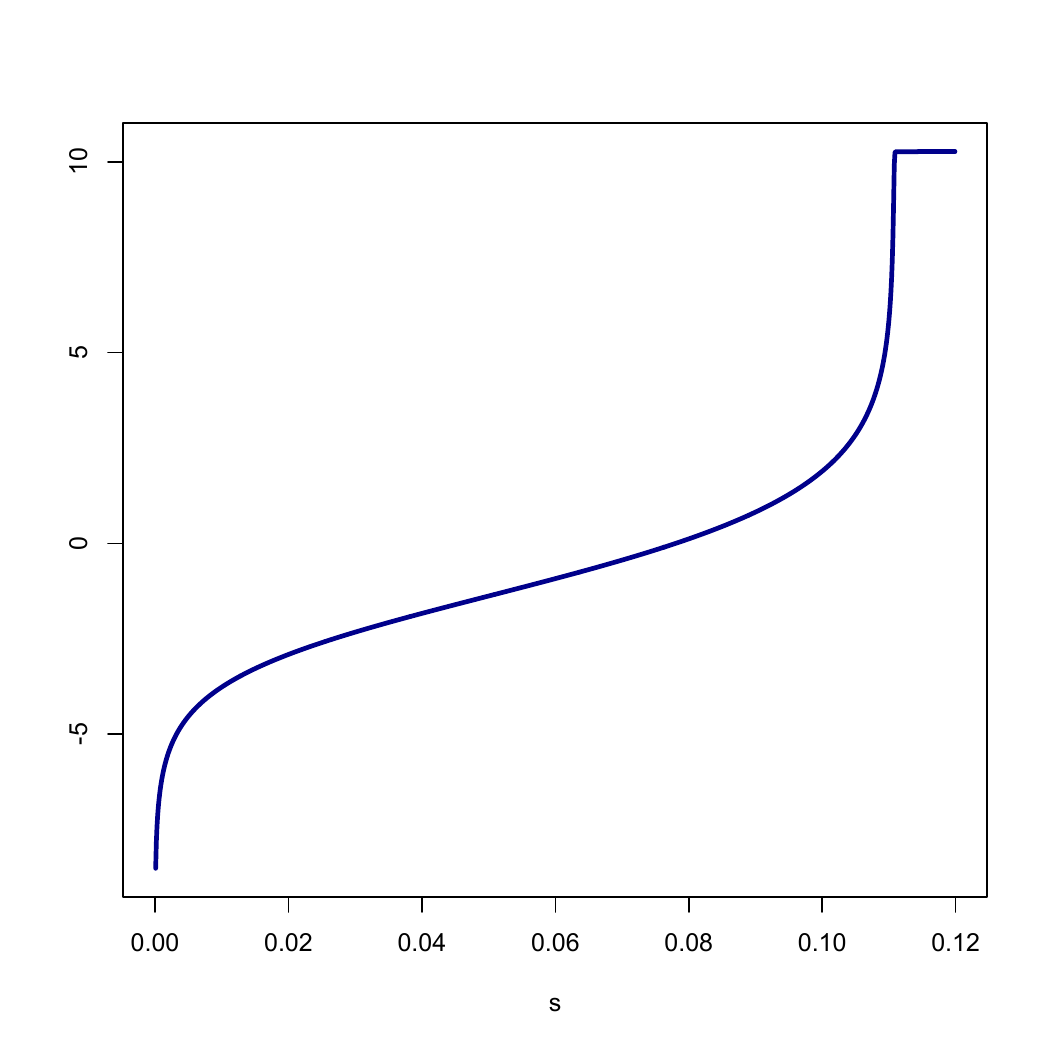}
\caption{\small Plots for $r_+ = 10.27$, $r_- = 6.5$, $R = 1.5$. \emph{Left:} Plots of $p^*(s)$ (= blue curve) and $\chi^*(s)$ (= red curve). \emph{Right:} Plot of $\ln[A^*_+(p^*(s))\chi^*(s)]$.}
\label{fig:reference-s}
\end{figure}

Figure~\ref{fig:reference-t} shows a plot of $\ln m(t)$ for $t \in [40,80]$ years based on \eqref{eq:micmac} and \eqref{eq:m*def}--\eqref{eq:scaldiff} for the values of $r_+,r_-,R,\Gamma_0$ used in \cite{MRFR2017}. The best fitted line as the linear regression fit is $\ln m(t) \approx -10.82 + 0.059\,t$, which gives
\[
\alpha \approx \eee^{-10.82} \approx 2.0 \times 10^{-5}, \qquad \beta \approx 5.9 \times 10^{-2}.
\]
The value of $\alpha$ is about $34\%$ of the value in Figure~\ref{fig:USmen}, the value of $\beta$ is about $69\%$ times the value in Figure~\ref{fig:USmen}. Thus, the match is fairly good.

\begin{Remark}
{\rm The linear regression fit depends on the $t$-interval that is chosen. For instance, Table \ref{table1} shows the results obtained by fitting linear regression lines to the data points on the curve $\ln m(t)$ at different time ranges.

\begin{center}
\begin{table}[htbp]
\centering
\caption{Properties of best fitted lines.}
\label{table1}
\renewcommand{\arraystretch}{1.6}
\begin{tabular}{cccccccccccc}
\toprule
\boldmath{\textbf{$t$-interval}} &&&& \boldmath{\textbf{$\alpha$}} && \boldmath{\textbf{$\beta$}} && \textbf{R-squared} && \textbf{MSE} \\
\midrule
$[40,80]$ years   &&&& $\approx 2.0 \times 10^{-5}$ && $\approx 5.9 \times 10^{-2}$ && 0.99 && 0.00 \\
$[40,85]$ years   &&&& $\approx 1.6 \times 10^{-5}$ && $\approx 6.3 \times 10^{-2}$ && 0.99 && 0.01 \\
$[40,90]$ years   &&&& $\approx 1.2 \times 10^{-5}$ && $\approx 6.8 \times 10^{-2}$ && 0.98 && 0.02 \\
$[40,95]$ years   &&&& $\approx 7.1 \times 10^{-6}$ && $\approx 7.8 \times 10^{-2}$ && 0.94 && 0.10 \\
$[40,100]$ years  &&&& $\approx 8.3 \times 10^{-7}$ && $\approx 1.1 \times 10^{-1}$ && 0.67 && 1.89 \\
\bottomrule
\end{tabular}
\end{table}
\end{center}


As observed in Table \ref{table1}, choosing a longer time interval decreases the value of $\alpha$ and increases the value of $\beta$. However, the decrease in $R$-squared and the increase in the mean-squared error MSE indicate that the fitted line is diverging from the curve. The reason is that the curve in Figure \ref{fig:reference-t} bends up. Thus, Theorem~\ref{th:Glaw} only yields a somewhat \emph{crude form} of the Gompertz law.}
\end{Remark}

\begin{Remark}
{\rm A numerical estimation of $\mathbb{E}[\hat{p}(\tau)]$ in \eqref{eq:frdeath} based on Monte-Carlo integration produces a value between $0.66$ and $0.70$ for the parameter values used in \cite{MRFR2017}. Hence, roughly and on average, a fraction $\frac23$ of the nodes is damaged when death occurs.}
\end{Remark}

\begin{figure}[htbp]
\centering
\includegraphics[width=0.425\textwidth]{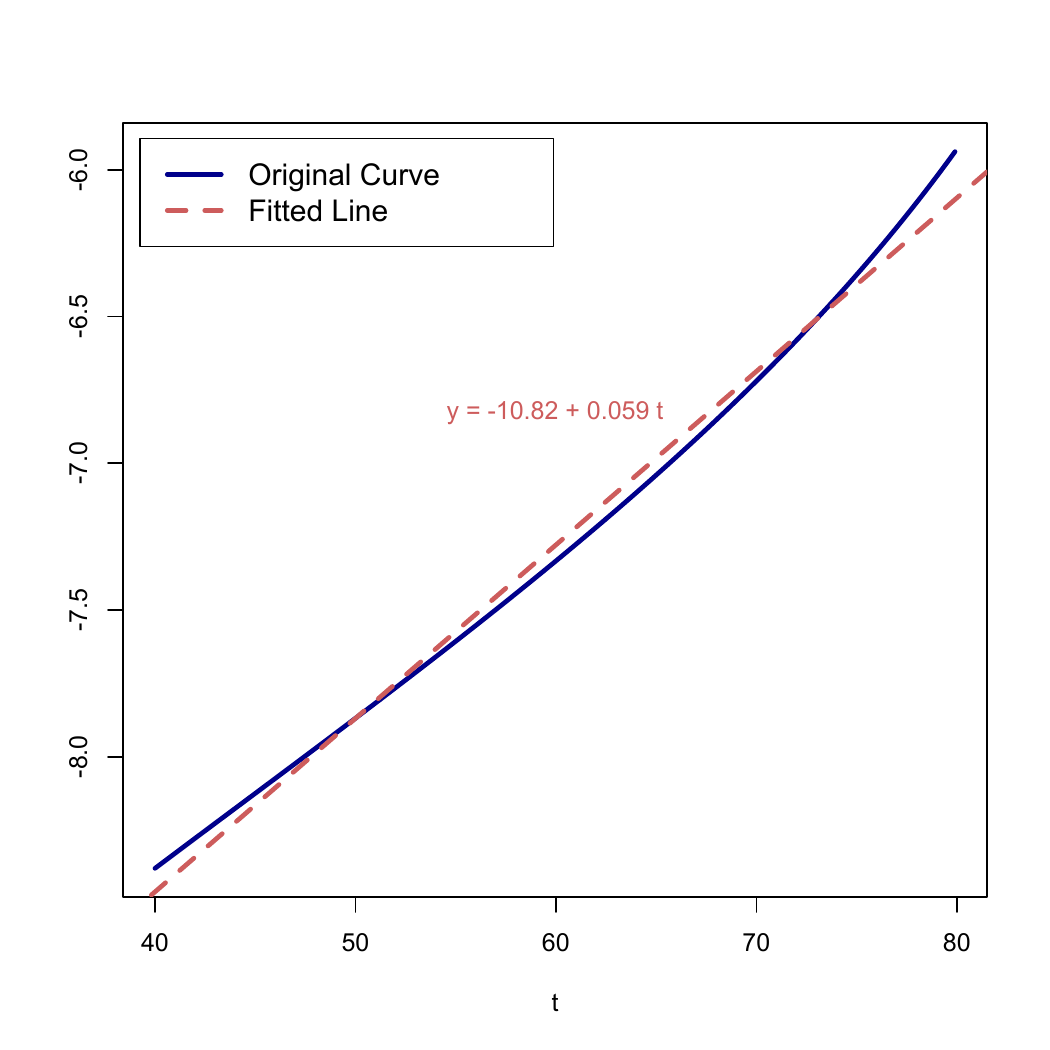}
\caption{\small Plot of $\ln m(t)$ for $t \in [40,80]$ years, for $r_+ = 10.27$, $r_- = 6.5$, $R = 1.5$, $\Gamma_0 = 0.00113$/year (= blue curve). The best fitted line (= dashed red line) corresponds to the Gompertz law.}
\label{fig:reference-t}
\end{figure}


\section{An analytic approximation}

The first differential equation in \eqref{eq:scaldiff} is non-linear and does \emph{not} admit a closed form solution for every choice of the parameters. The second differential equation in \eqref{eq:scaldiff} is linear given the solution of the first and can therefore be solved explicitly:
\[
\chi^*(s) = \int_0^s \ddd u\,2A^*_+(p^*(u))\,\exp\left[-\int_u^s \ddd v\,[2A^*_+(p^*(v)) + A^*_-(p^*(v))]\right].
\]
Both differential equations are not hard to solve \emph{numerically}, as shown in Section~\ref{sec:simulations}, but it is natural to ask whether it is possible to find an \emph{analytic approximation} of $p^*(s)$, $\chi^*(s)$ and $\ln[A^*_+(p^*(s))\,\chi^*(s)]$. Below we give an affirmative answer. Any formula that expresses these quantities directly in terms of the parameters $r_+,r_-,R$, even in an approximate form, is helpful for a better understanding of the evolution of the damage fraction and the mortality rate.

The differential equation for $p^*(s)$ in \eqref{eq:scaldiff} can be integrated to give
\[
s = \int_0^{p^*(s)} \frac{\ddd\lambda}{(1-\lambda)\,A^*_+(\lambda) - \lambda\,A^*_-(\lambda)},
\]
where we use that $p^*(0)=0$, noting that the integrand is integrable near $\lambda = 0$. In order to obtain a formula for $p^*(s)$ as a function of $s$, the expression needs to be inverted. We proceed by deriving upper and lower bounds.

\medskip\noindent
\underline{Lower bound:} Estimate
\[
\begin{aligned}
s &= \int_0^{p^*(s)} \frac{\ddd\lambda\,\eee^{-r_+\lambda}}{(1-\lambda) - \frac{\lambda}{R}\,\eee^{-(r_+ + r_-)\lambda}}
\leq \int_0^{p^*(s)} \frac{\ddd\lambda\,\eee^{-r_+\lambda}}{(1-\lambda) - \frac{\lambda}{R}}\\
&\leq \frac{1}{1-\frac{R+1}{R}\,p^*(s)} \int_0^{p^*(s)} \ddd\lambda\,\eee^{-r_+\lambda}
=  \frac{1}{1-\frac{R+1}{R}\,p^*(s)}\,\frac{1}{r_+}\left[1-\eee^{-r_+ p^*(s)}\right]
\leq \frac{p^*(s)}{1-\frac{R+1}{R}\,p^*(s)}.
\end{aligned}
\]
Inverting this inequality, we get
\[
p^*(s) \geq \frac{s}{1+\frac{R+1}{R}\,s}.
\]
Since $\chi^*(s) \geq p^*(s)$ by Lemma~\ref{lem:sandwich}, it follows from \eqref{eq:m*def} that
\[
m^*(s) \gtrapprox \frac{\Gamma_0\,s}{1+\frac{R+1}{R}\,s}
\exp\left[\frac{r_+ s}{1+\frac{R+1}{R}\,s}\right],
\]
which via \eqref{eq:micmac} becomes
\begin{equation}
\label{eq:lb}
m(t) \gtrapprox \frac{\Gamma_0^2\,t}{1+\frac{R+1}{R}\,\Gamma_0\,t}
\exp\left[\frac{r_+ \Gamma_0\,t}{1+\frac{R+1}{R}\,\Gamma_0\,t}\right],
\end{equation}
In case $\frac{R+1}{R}\,\Gamma_0\,t \ll 1$, the right-hand side simplifies to $\Gamma_0^2\,t\,\eee^{r_+\Gamma_0\,t}$.

\medskip\noindent
\underline{Upper bound:} Estimate
\[
s \geq \int_0^{p^*(s)} \ddd\lambda\,\eee^{-r_+\lambda} = \frac{1}{r_+}\left[1-\eee^{-r_+ p^*(s)}\right].
\]
Inverting this inequality for $s$ that $r_+ s<1$, we get
\[
p^*(s) \leq - \frac{1}{r_+} \ln(1-r_+ s) = \frac{1}{r_+} \sum_{k=1}^\infty \frac{(r_+ s)^k}{k}.
\]
Since $\chi^*(s) \leq 2 p^*(s)$ by Lemma~\ref{lem:sandwich}, it follows from \eqref{eq:m*def} that
\[
m^*(s) \lessapprox \frac{2\Gamma_0}{r_+} \left[\sum_{k=1}^\infty \frac{(r_+ s)^k}{k}\right]
\exp\left[\sum_{k=1}^\infty \frac{(r_+ s)^k}{k}\right],
\]
which via \eqref{eq:micmac} becomes
\begin{equation}
\label{eq:ub}
m(t) \lessapprox \frac{2\Gamma_0}{r_+} \left[\sum_{k=1}^\infty \frac{(r_+ \Gamma_0\,t)^k}{k}\right]
\exp\left[\sum_{k=1}^\infty \frac{(r_+ \Gamma_0\,t)^k}{k}\right].
\end{equation}
In case $r_+\Gamma_0\,t \ll 1$, the right-hand side simplifies to $2\Gamma_0^2\,t\,\eee^{r_+\Gamma_0\,t}$.

The bounds in \eqref{eq:lb}--\eqref{eq:ub} show that, as long as both $\frac{R+1}{R}\Gamma_0\,t \ll 1$ and $r_+ \Gamma_0\,t \ll 1$, a fair approximation for $\beta$ in the Gompertz law is
\begin{equation}
\label{eq:betaapprox}
\beta \approx r_+\Gamma_0,
\end{equation}
while $\alpha$ can be crudely sandwiched for $t\in[40,80]$ as
\begin{equation}
\label{eq:alphapprox}
40\,\Gamma_0^2 \lessapprox \alpha \lessapprox 160\,\Gamma_0^2.
\end{equation}
The bounds in \eqref{eq:alphapprox} are a factor $4$ apart.

\begin{figure}[htbp]
\centering
\includegraphics[width=0.425\textwidth]{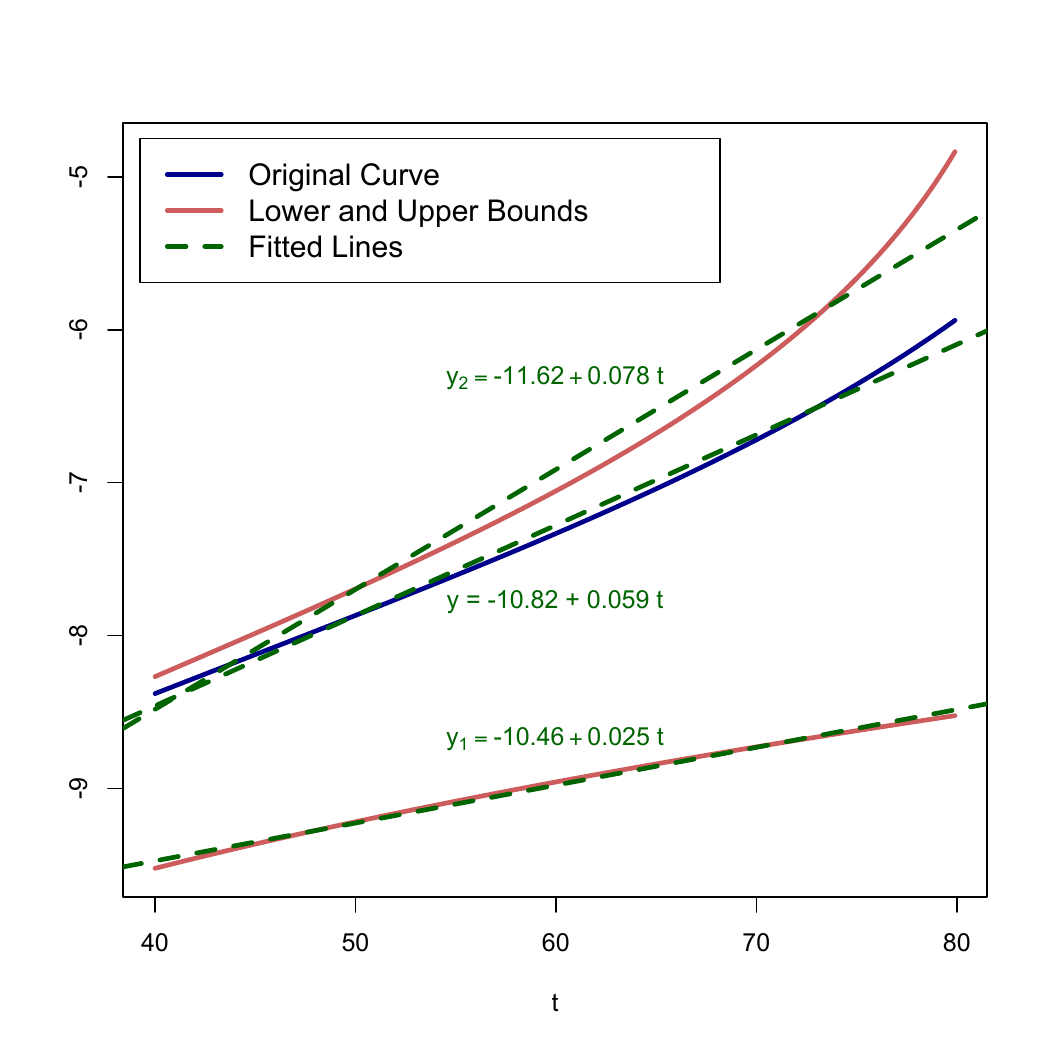}
\caption{\small Plot of $\ln m(t)$ for $t \in [40,80]$ years for the parameter values used in \cite{MRFR2017}. The middle curve (= blue curve) is the one in Figure~\ref{fig:reference-t}. The bottom curve and the top curve (= red curves) represent the numerical plot of the logarithm of the lower bound in \eqref{eq:lb} and the upper bound in \eqref{eq:ub}. The best fitted lines (= green dashed curves) correspond to the Gompertz law.}
\label{fig:reference-talt}
\end{figure}

For the parameter values used in \cite{MRFR2017}, Figure \ref{fig:reference-talt} compares the curve in Figure~\ref{fig:reference-t} with the bounds obtained in \eqref{eq:lb}--\eqref{eq:ub}. It seems that the upper bound in \eqref{eq:ub} is better than the lower bound in \eqref{eq:lb}. The slopes of the best fitted lines are $7.8 \times 10^{-2}$ for the upper bound, $2.5 \times 10^{-2}$ for the lower bound, and $5.9 \times 10^{-2}$ for the true curve. Thus, the values of $\beta$ obtained from the bounds are off by about 50\%, which is fair but not sharp. Interestingly, the values of $\ln \alpha$ in the three cases differ by not more than 8\%, which is rather sharp.

For the parameter values used in \cite{MRFR2017}, $\frac{R+1}{R}\Gamma_0\,t \approx 0.075$ and $r_+ \Gamma_0 t \approx 0.46$ at $t=40$ years. The former is small, the latter is not, which makes the approximations in \eqref{eq:betaapprox}--\eqref{eq:alphapprox} a priori questionable. The approximation in \eqref{eq:betaapprox} would predict $\beta \approx 1.2 \times 10^{-2}$, which is below the true value by a factor of about 5. The approximation in \eqref{eq:alphapprox} would predict $-9.88 \lessapprox \ln \alpha \lessapprox -8.50$, which interval lies above the true value. The lower bound of the interval is off by about 9\%, the upper bound by about 20\%. Again, the approximation of $\alpha$ is better than that of $\beta$.


\section{Conclusion}

The main contributions of the present paper are the following.
\begin{itemize}
\item[(1)]
We have provided \emph{mathematical} arguments to support the \emph{network} description of aging and mortality proposed in \cite{MRFR2017}.
\item[(2)]
We have argued that Poisson rates for the evolution of the network are \emph{universal} as long as the mechanism according to which nodes switch between healthy and damaged is ``the net result of many small influences''.
\item[(3)]
With the help of two assumptions, valid for \emph{scale-free} and \emph{disassortative} networks, we have shown that the \emph{Gompertz law} holds in a certain sense in a certain age range. The precise formula for the mortality rate \emph{deviates} somewhat from the Gompertz law. With the help of simulations we have shown that the approximations implied by these assumptions are fair. The size of the network plays no role after these approximations have been implemented, which again underlines the \emph{universality} of the network description.
\item[(4)]
We have linked the parameters in the Gompertz law to the parameters driving the evolution of the network. For the curves in Figures~\ref{fig:USmen} and \ref{fig:numplot} we have found a \emph{fair fit} with the model parameters used in \cite{MRFR2017}.
\item[(5)]
Our formula for the mortality rate involves the solution of two non-linear differential equations, which are easy to handle numerically but are hard to solve analytically. Exact bounds on the solutions of these differential equations allow us to derive a \emph{crude} analytic approximation of the mortality rate that is \emph{explicit} in the model parameters.
\end{itemize}

It is well known that the Gompertz law is not perfect. Yet, it is ubiquitous, simple and widely regarded as a good approximation of the \emph{empirical law} for the mortality rate as a function of age \cite{OC1997}. The Gompertz law is known to be relevant in many species, from yeast to fruit flies, from dogs to horses. For instance, mice accumulate health deficits just like humans \cite{RBTSFMH2017}. Thus, the network description of ageing via the accumulation of health deficits \emph{in principle} covers a wide range of species, with parameters varying across species.


\vskip 1 cm \noindent
\emph{Acknowledgement.}
FdH and AP were supported by the Netherlands Organisation for Scientific Research (NWO) through NETWORKS Gravitation Grant no.\ 024.002.003. AP was also supported by the European Union's Horizon 2020 research and innovation programme under the Marie Sk\l odowska-Curie grant agreement no.\ 101034253.

\vspace{0.3cm}\hspace{-0.3cm}
 \includegraphics[height=3em]{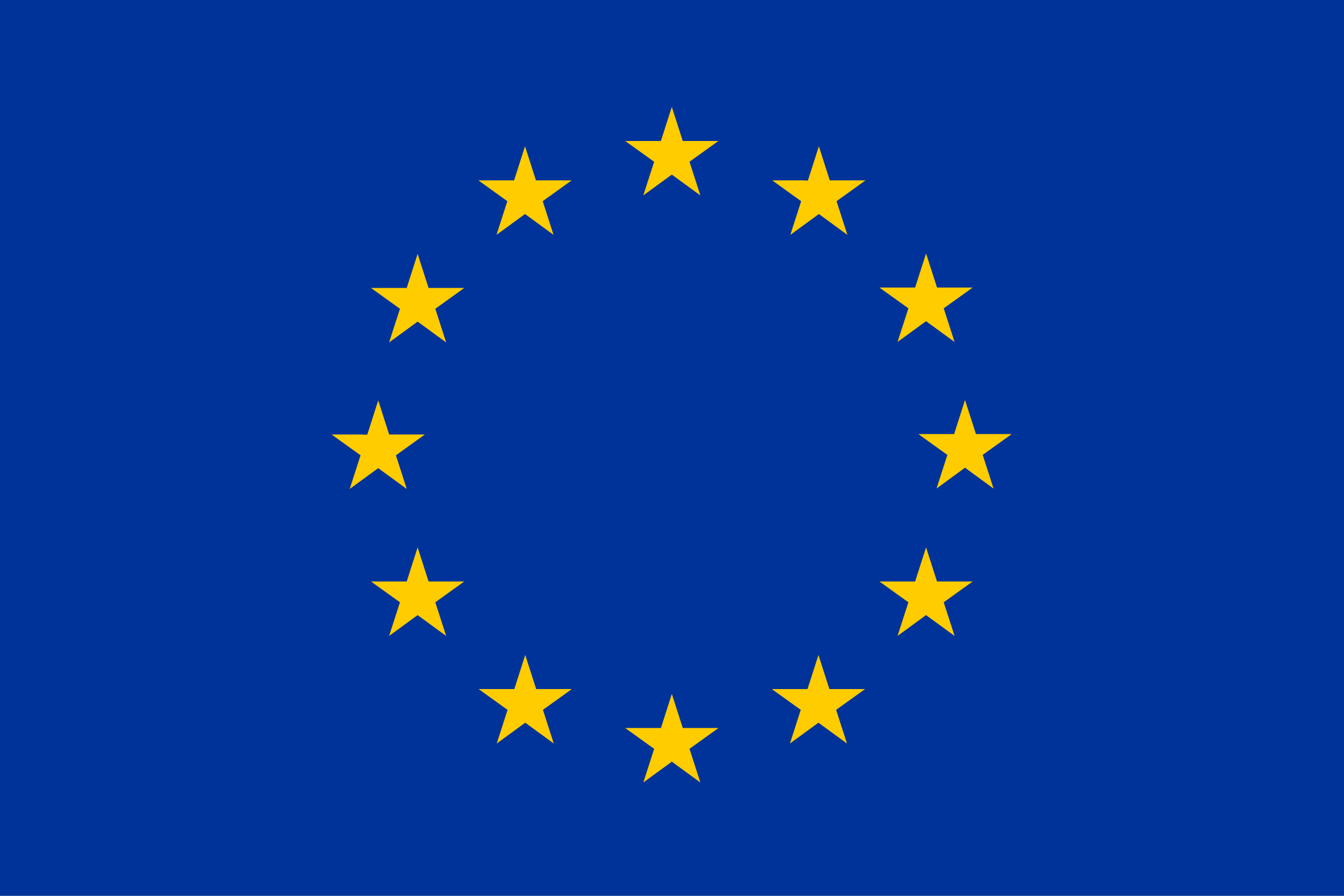}

\clearpage


\end{document}